\documentclass[11pt]{article}

\usepackage{graphicx} % Required for inserting images
\usepackage[utf8]{inputenc}
\usepackage{eurosym,geometry,graphicx,color,setspace,sectsty,comment,footmisc,natbib,pdflscape,array,hyperref}
\usepackage[normalem]{ulem}
\usepackage{caption}
\usepackage{subcaption}
\usepackage{amssymb,amsfonts,amsmath}
\usepackage{amsthm}
\usepackage{natbib}
\usepackage{booktabs}
\usepackage{threeparttable}
\usepackage{tikzsymbols} % emoji
\usepackage{tikz}
\usetikzlibrary{patterns,arrows}
\usepackage{apxproof}
\usepackage{siunitx}
\usepackage{xr}
\usepackage{xcolor} % Colored text
\usepackage{todonotes}
\setuptodonotes{inline}
\usepackage[]{mdframed}

\newtheorem{assumptionD}{Assumption}

\newtheorem{assumptionM}{Assumption}

\newtheorem{lemma}{Lemma}

\newtheorem{proposition}{Proposition}

\theoremstyle{definition} % Upright text

\newtheorem{example}{Example}
\newtheorem*{example*}{Example}

\AtBeginEnvironment{example}{%
  \pushQED{\qed}%
}
\AtEndEnvironment{example}{\popQED\endexample}

\hypersetup{colorlinks,citecolor=blue} % citation in blue
\title{Merger Analysis with Unobserved Prices\footnote{I thank John Asker, Matthew Backus, Chris Conlon, Diego Cussen, Xiao Dong, Karam Kang, Marc Luppino, Chris Metcalf, Nathan Miller, Scott Orr, Devesh Raval, Ted Rosenbaum, David Schmidt, Gloria Sheu, Eddie Watkins, Brett Wendling, two anonymous referees for the FTC Working Paper Series, and the participants of the workshops at the Federal Trade Commission, Korea Fair Trade Commission, ITAM, and 2024 IIOC in Boston for helpful discussion and comments. Much of this article was written during the author’s tenure at the U.S. Federal Trade Commission; the views expressed in this article are those of the author and do not necessarily represent those of the Federal Trade Commission or any of its Commissioners. This article was originally circulated as ``Merger Analysis with Latent Price.'' All errors are mine.} }
\author{Paul S. Koh\footnote{School of Economics, Yonsei University. 50, Yonsei-ro, Seodaemun-gu, Seoul, Republic of Korea. Email: \texttt{paulkoh9@gmail.com}.}}
\date{August 30, 2025}

\begin{document}

\maketitle
\begin{abstract}
Standard empirical tools for merger analysis assume price data, which are often unavailable. I characterize sufficient conditions for identifying the unilateral effects of mergers without price data using the first-order approach and merger simulation. Data on merging firms' revenues, margins, and revenue diversion ratios are sufficient to identify their gross upward pricing pressure indices and compensating marginal cost reductions. Standard discrete-continuous demand assumptions facilitate the identification of revenue diversion ratios as well as the feasibility of merger simulation in terms of percentage change in price. I apply the framework to the Albertsons/Safeway (2015) and Staples/Office Depot (2016) mergers.
    
\end{abstract} \clearpage
\section{Introduction}

\subsection{Motivation}
A horizontal merger is said to generate unilateral effects when it reduces competition between the merging parties and enables them to increase prices. Standard empirical tools for evaluating mergers typically assume that price (and quantity) data are available \citep{davis2009quantitative, valletti2021mergers, miller2021quantitative}. Merger simulations use price data to estimate or calibrate demand functions. The first-order approach that measures unilateral effects via upward pricing pressures and compensating marginal cost reductions requires price data as either direct or indirect inputs to the formulas and for estimating demand slopes and curvatures.

%%% Original %%%
% However, researchers often have difficulty accessing reliable price and quantity data. First, price and quantity data may not exist. Second, even if the data exists, the researcher may not have access. Third, even if the researcher can access the data, it may be costly to process it. For example, if a merger involves retail firms (e.g., supermarket chains) whose outlets carry thousands of non-overlapping items with prices that frequently vary due to complex promotion activities, constructing relevant price indices can be non-trivial and time-consuming. The lack of ``clean'' price data often creates considerable challenges for economists and antitrust agencies in predicting the price and welfare effects of mergers.

%%% Revised %%%
However, researchers often face significant challenges in accessing reliable price and quantity data. In some cases, the data may not exist due to limited record-keeping or firms’ unwillingness to share competitively sensitive data. Furthermore, even if the researcher can access the data, it may be costly to process it. As a result, antitrust economists frequently lack the price and quantity data required for standard demand and merger analysis. While tools like quantity diversion ratios can still be used without price data when quantity data are available, predicting price effects remains difficult without additional information.

\subsection{Main Findings}
This article develops an empirical framework for estimating the unilateral effects of horizontal mergers without price data, extending the scope of applicability of the first-order approach and merger simulation. I adopt the standard empirical Bertrand-Nash multiproduct oligopoly framework. Unlike conventional merger analysis, however, I study the identifiability of unilateral effects using revenue and margin data---metrics that firms routinely track in their accounting and managerial practices---while treating prices and quantities as unobserved.\footnote{Firms routinely track revenues and margins in the ordinary course of their accounting and managerial practices because these measures are central to financial reporting, tax compliance, and internal performance evaluation. Public companies disclose them regularly in annual reports, SEC filings, and earnings calls. Industry trade publications and analyst reports also tend to cite market shares in revenue terms and margin comparisons across firms, rather than price and quantity distributions.} By doing so, my framework bypasses the usual assumption that prices and quantities must be observed separately, offering a feasible approach to merger analysis when collecting such data is costly.

My key idea is to focus on \emph{revenue diversion ratios} instead of \emph{quantity diversion ratios}. Whereas quantity diversion ratios, which track how consumers shift in terms of units purchased, have long been used to assess product substitutability since \citet{shapiro1995mergers}, revenue diversion ratios, which reflect shifts in consumer spending, have received little attention. I describe the firm-side model and definition of revenue diversion ratios in Section \ref{section:2.model}.

In Section \ref{section:3.identification.of.unilateral.effects}, I show that data on merging parties' revenues, margins, and revenue diversion ratios are sufficient to identify their gross upward pricing pressure indices and compensating marginal cost reductions. Using the definition of revenue diversion ratios, I can express firms' first-order conditions and gross upward pricing pressure statistics as functions of their margins, revenue diversion ratios, and own-price demand elasticities. In turn, the own-price demand elasticities of products can be identified from the owner firm's margins and revenue diversion ratios data via the profit maximization condition.\footnote{The standard empirical approach uses estimates of price elasticities to identify margins via the profit maximization conditions. Here, I reverse their roles and use margins to identify own-price elasticities.} Thus, the gross upward pricing pressure statistics are identifiable from the revenues, margins, and revenue diversion ratios of the merging firms. The analyst can translate gross upward pricing pressure indices to first-order merger price and welfare effects if pass-through rates are available \citep{jaffe2013first}. Analogous arguments show that the same data assumption identifies compensating marginal cost reductions. 

Like quantity diversion ratios, revenue diversion ratios can be estimated using internal firm documents, quasi-experiments, surveys, or parametric assumptions, with revenues as the outcome variable. I show that standard discrete-continuous demand assumptions facilitate identification, much like how discrete choice models aid in identifying quantity-based diversion. Rather than assuming unit-inelastic demand, I model consumers as allocating their budgets by making a discrete choice for each unit of expenditure. Under the standard additive random utility framework, this yields expenditure shares consistent with the GEV class and allows application of familiar tools from the discrete choice literature. When consumers have a homogeneous responsiveness to price---a strong but standard assumption in merger analysis---revenue diversion ratios can be identified using \citet{hotz1993conditional}'s inversion, which maps observed market shares to latent utilities. I illustrate this in Section \ref{section:4.consumer.demand} using CES preferences, where expenditure shares take a multinomial logit (softmax) form, enabling identification from consumer spending to merging firms' products or second-choice data. Appendix \ref{subsection:revenue.diversion.ratios.under.discrete.continuous.demand} extends the approach to a broader class of discrete-continuous demand models.

Next, I turn to merger simulation and provide sufficient conditions for conducting it without price data, thereby complementing the existing literature. Under a discrete-continuous demand assumption, I demonstrate that while post-merger price levels are not identifiable, the percentage change in prices relative to pre-merger levels is. The key is to express post-merger pricing conditions as known functions of percentage price deviations, which is possible under the assumption of log-linearity in price. However, merger simulation requires observing all competitors' revenues and margins, making it more data-intensive than the first-order approach. In practice, the analyst can recover non-merging firms' margins using demand assumptions and data on merging firms' margins. Section \ref{section:5.merger.simulation} illustrates the implementation using CES preferences, and Appendix \ref{subsection:merger.simulation.under.discrete.continuous.demand} extends the logic to a broader class of discrete-continuous models.

I illustrate the usefulness of my methodology with two empirical applications, both of which show how revenue-based measures can ease the burden of collecting product-level price and quantity data for firms with thousands of products, by aggregating sales into a single revenue index at the firm or store level and treating the underlying price and quantity components as latent. First, in Section \ref{section:6.empirical.application}, I apply the framework to the Staples/Office Depot merger, which was proposed in 2014 but eventually blocked in 2016. This empirical example is chosen to illustrate the simplicity of my approach. In its complaint, the FTC claimed a cluster market for consumable office supplies that includes many individual items. The nature of the claimed market makes it challenging to construct price data that fit standard empirical frameworks. My framework does not require price data. Using publicly available data on the firms' revenues, margins, and market shares, I evaluate the merger's unilateral effects using the first-order approach and merger simulation.  All approaches predict either the merger produces substantial harm or requires significant cost reductions to offset upward pricing pressures. 

Next, in Section \ref{section:7.empirical.application.II}, I apply the framework to analyze the Albertsons/Safeway merger, the largest grocery merger in US history. This empirical example illustrates the scalability of my approach. The Federal Trade Commission approved the merger in 2015, conditional on divesting 168 stores. I estimate a spatially aggregated nested-CES demand model using cross-sectional data on store revenues to estimate revenue diversion ratios. I then calculate store-level merger price/welfare effects for thousands of stores before and after the divestiture. I find that the FTC-mandated divestiture significantly reduced annual consumer harm. I also estimate the distribution of cost efficiencies required to offset upward pricing pressures. My results inform how the FTC assigned cost-efficiency credits to the parties' stores during the merger investigation.
\subsection{Relationship to the Literature}
This article primarily contributes to the merger analysis literature by developing an empirical framework for analyzing horizontal mergers without price data. It offers four main contributions relative to prior work. First, it complements the literature on empirical merger tools based on easily computed sufficient statistics \citep{werden1996robust, farrell2010antitrust, jaffe2013first, affeldt2013upward, miller2013using, weyl2013pass, brito2018unilateral, miller2021quantitative}. Building on the standard Bertrand-Nash framework, I provide a novel identification strategy for estimating first-order unilateral effects using \emph{revenue diversion ratios}---a topic that has received little attention, with the exception of \citet{caradonna2023mergers}, who study merger-induced entry under (nested) multinomial and CES demand. My identification arguments relax the data requirements in \citet{werden1996robust}, \citet{farrell2010antitrust} and \citet{jaffe2013first}, who assumed prices and quantities are observed separately (along with margins). 

Second, this article contributes to a large body of research on merger simulation \citep{hausman1994competitive, werden1994effects, nevo2000mergers, epstein2001merger, werden2002antitrust}. I show that I can calculate percentage changes in price without estimating a complete set of demand function parameters. My results are similar to those of \citet{miller2014modeling, miller2017modeling}, which show that percentage changes in markups can be identified from merging firms' market shares in procurement settings.

Third, this article contributes to a body of research that finds the value of employing discrete-continuous demand assumptions for theoretical and empirical merger analysis \citep{anderson1987ces, anderson1988ces, anderson1992discrete, bjornerstedt2016does, nocke2018multiproduct, nocke2022concentration, taragin2022antitrust, caradonna2023mergers, nocke2023aggregative, garrido2024aggregative}. Employing a discrete-continuous demand assumption facilitates the identification of revenue diversion ratios and merger simulation when price data are absent. Specifically, assuming consumers' expenditures to products are generated from an additive random utility model for every unit of budget allows the analyst to apply the standard econometric results from the GEV discrete choice literature. While a common approach to calculating gross upward pricing pressure indices in the absence of price data has been to use revenues as proxies for quantities or to assume that prices of merging firms' products are approximately equal \citep{ferguson2023economics}, I show that such ad hoc assumptions are unnecessary and may generate substantial bias. This article also relates to a body of research that recognizes how the CES demand assumption allows the analyst to estimate production functions with revenue data when prices and quantities are unobserved \citep{klette1996inconsistency, de2011product, grieco2016production, gandhi2020identification}.

Finally, my empirical applications complement the existing body of research on retail merger retrospectives and divestiture remedies \citep{smith2004supermarket, hosken2016can, allain2017retail, thomassen2017multi, hosken2018retail, ellickson2020measuring}. My empirical applications are new in the literature. The closest to my work are \citet{smith2004supermarket} and \citet{ellickson2020measuring}, which also study supermarket mergers. \citet{smith2004supermarket} also uses profit margins data and equilibrium pricing conditions to analyze supermarket competition but develops a structural model tailored to fit consumer shopping patterns data in the supermarket industry. In contrast, my framework follows the standard empirical Bertrand-Nash pricing model. \citet{ellickson2020measuring} develops a spatial demand framework that overcomes the absence of price data in the grocery competition setting but does not directly calculate merger price effects. My approach can calculate merger price effects under the same set of assumptions.

To clarify its scope, my framework is most useful for analyzing standard Bertrand-Nash competition settings where the analyst lacks separate price and quantity data but observes revenues and margins. Retail mergers---such as the recent Kroger/Albertsons merger \citep{ftc2024kroger}, which resembles my Albertsons/Safeway empirical example---are well-suited applications. The framework also helps guide analysts on the types of data needed for merger analysis, showing that traditional price and quantity data may be bypassed if percentage price changes can be used in place of price levels. The analyst can choose between the standard framework and my framework depending on the data availability scenario. I do not claim that estimating revenue diversion ratios is always feasible or inherently simpler than estimating quantity diversion ratios, but only that they are often more tractable when standard tools are impractical due to data constraints. While price data is often unavailable in various settings, this paper focuses on cases where Bertrand-Nash competition is a reasonable approximation. The framework does not extend to other important models---such as those involving bargaining, auctions, nonlinear or dynamic pricing, or collusion---which merit further study but are beyond the scope of this work.
\subsection{Outline}

The rest of the article is organized as follows. In Section \ref{section:2.model}, I describe the firm-side model and introduce the concept of revenue diversion ratios. In Section \ref{section:3.identification.of.unilateral.effects}, I establish identification conditions for gross upward pricing pressure indices and compensating marginal cost reductions. In Section \ref{section:4.consumer.demand}, I characterize consumer demand assumptions that facilitate the estimation of revenue diversion ratios. In Section \ref{section:5.merger.simulation}, I show merger simulation is feasible. In Sections \ref{section:6.empirical.application} and \ref{section:7.empirical.application.II}, I apply the proposed methodology to evaluate the Staples/Office Depot merger (2016) and the Albertsons/Safeway merger (2015). Finally, I conclude in Section \ref{section:99.conclusion}. All proofs are in Appendix \ref{section:proof}. Online Appendix can be found at the \href{https://www.pskoh.com}{author's website}.

\section{Model \label{section:2.model}}
In this section, I describe the firm-side model. I also introduce the concept of revenue diversion ratio and discuss its relationship with quantity diversion ratios.
\subsection{Setup}
The firm-side model primitives are summarized by a tuple $\langle \mathcal{J}, \mathcal{F}, (c_j)_{j \in \mathcal{J}} \rangle$, where $\mathcal{J}$ is the set of products, $\mathcal{F}$ is the set of multiproduct firms, and $c_j \in \mathbb{R}_+$ specifies the product $j$'s (constant) marginal cost. Set $\mathcal{F}$ forms a partition over $\mathcal{J}$, specifying firms' ownership over products. The set of products owned by firm $F \in \mathcal{F}$ is denoted $\mathcal{J}_F \subseteq \mathcal{J}$. I assume the profit from product $j \in \mathcal{J}$ is $\pi_j = (p_j - c_j) q_j$ where $p_j \in \mathbb{R}_+$ and $q_j = q_j(p)$ denote price and quantity demanded, respectively; I omit fixed costs for notational convenience. I also assume that the products are substitutes. I use $m_j \equiv (p_j - c_j)/p_j$ to denote relative margins.
\subsection{Firm's Problem}
The multiproduct firms engage in a Bertrand-Nash pricing game. Each firm $F \in \mathcal{F}$ maximizes its total profit $\sum_{j \in \mathcal{J}_F} \pi_j$ with respect to a vector of prices $(p_j)_{j \in \mathcal{J}_F}$. Normalizing the first-order conditions to be quasilinear in margins yields
\begin{equation}\label{equation:optimal.pricing.equation.2}
    -\epsilon_{jj}^{-1} - m_j + \sum_{l \in \mathcal{J}_F \backslash j} m_l D_{j \to l} \frac{p_l}{p_j} = 0,
\end{equation}
where $\epsilon_{jj} \equiv \frac{\partial q_j}{\partial p_j} \frac{p_j}{q_j}$ is the own-price elasticity of demand, and 
\begin{equation}\label{equation:quantity.diversion.ratio}
    D_{j \to l} \equiv - \frac{\partial q_l / \partial p_j}{\partial q_j / \partial p_j}
\end{equation}
is the quantity diversion ratio from product $j$ to product $l$. I defer the review and derivations of the firms' optimal pricing equations, gross upward pricing pressure indices, and compensating marginal cost reductions to Online Appendix \ref{section:B.review.of.unilateral.effects.analysis}, as these are standard results.
\subsection{Revenue Diversion Ratios and Their Properties}
\paragraph*{Definition}
Standard empirical frameworks for mergers' unilateral effects analysis have focused on \emph{quantity diversion ratios} \eqref{equation:quantity.diversion.ratio} as the key statistics for measuring substitutability between products \citep{shapiro1995mergers, farrell2010antitrust, conlon2021empirical}. To establish the identifiability of merger price effects without price data, I will rewrite price elasticities and gross upward pricing pressure indices in terms of \emph{revenue diversion ratios}. Revenue diversion ratio from product $j$ to $k$ is defined as
\begin{equation}\label{equation:revenue.diversion.ratio}
    D_{j \to k}^R \equiv - \frac{\partial R_k / \partial p_j}{\partial R_j / \partial p_j},
\end{equation}
where $R_l\equiv p_l \cdot q_l(p)$ is product $l$'s revenue. It measures the substitutability between two products by studying how \emph{revenue} shifts from one product to another following a unilateral price increase. Under standard regularity conditions, revenue diversion ratios are non-negative in equilibrium.\footnote{First, $\partial R_k / \partial p_j \geq 0$ if $j$ and $k$ are substitutes. Second, $\partial R_j / \partial p_j <0$ if and only if $\epsilon_{jj} < -1$, which holds in any Bertrand-Nash equilibrium. In sum, $D_{j \to k}^R > 0$. The only exception is the revenue diversion ratio for a self-pair $D_{j \to j}^R \equiv -1 < 0$.} 

\paragraph*{Relationship to Quantity Diversion Ratios}
Revenue and quantity diversion ratios are different, but they are closely related. Let $\epsilon_{jj} =  \frac{\partial q_j}{\partial p_j} \frac{p_j}{q_j}$ and $\epsilon_{jj}^R =  \frac{\partial R_j}{\partial p_j} \frac{p_j}{R_j}$ denote the own-price elasticity of demand and the own-price elasticity of revenue, respectively. Similarly, let $\epsilon_{kj} =  \frac{\partial q_k}{\partial p_j} \frac{p_j}{q_k}$ and $\epsilon_{kj}^R =  \frac{\partial R_k}{\partial p_j} \frac{p_j}{R_k}$ denote the cross-price elasticities. The following lemma summarizes their relationship.

\begin{lemma}[Relationship between revenue-based and quantity-based measures] \label{lemma:relationship} 
For an arbitrary pair of products $j$ and $k$, 
\begin{enumerate}
    \item \label{lemma:relationship.item.1} $D_{j \to k}^R = - \frac{\epsilon_{kj}^R}{\epsilon_{jj}^R} \frac{R_k}{R_j}$ and $D_{j \to k} = - \frac{\epsilon_{kj}}{\epsilon_{jj}} \frac{q_k}{q_j}$;
    \item \label{lemma:relationship.item.2} $\epsilon_{jj}^R = \epsilon_{jj} + 1$, and $\epsilon_{kj}^R = \epsilon_{kj}$ for $j \neq k$;
    \item \label{lemma:relationship.item.3} $(1 + \epsilon_{jj}^{-1}) D_{j \to k}^R = D_{j \to k} \frac{p_k}{p_j}$ for $j \neq k$, and $D_{j \to j}^R = D_{j \to j} \equiv -1$.
\end{enumerate}
\end{lemma}
Lemma \ref{lemma:relationship} uses the algebraic definition of diversion ratios and elasticities and is thus independent of the underlying demand model. Lemma \ref{lemma:relationship}.\ref{lemma:relationship.item.1} relates diversion ratios to own-/cross-price elasticities. Lemma \ref{lemma:relationship}.\ref{lemma:relationship.item.2} shows how the own-/cross-price elasticities are related. Finally, Lemma \ref{lemma:relationship}.\ref{lemma:relationship.item.3} shows how quantity diversion ratios may be substituted out for revenue diversion ratios.

Lemma \ref{lemma:relationship}.\ref{lemma:relationship.item.3} plays a key role in my identification arguments. The terms $D_{j \to k} \frac{p_k}{p_j}$ enter the firms' optimal pricing equation, gross upward pricing pressure indices, and compensating marginal cost reductions. Conventional approaches assume price and quantity data to calculate them. However, for $j \neq k$,
\[
D_{j \to k}\frac{p_k}{p_j} = - \frac{ \frac{\partial q_k}{\partial p_j} p_k}{ \frac{\partial q_j}{\partial p_j} p_j} =- \frac{\frac{\partial R_k}{\partial p_j}}{\frac{\partial R_j}{\partial p_j} - q_j} = \underbrace{\left( -\frac{\frac{\partial R_k}{\partial p_j}}{\frac{\partial R_j}{\partial p_j}}\right)}_{D_{j \to k}^R} 
\underbrace{ \left( \frac{\frac{\partial R_j}{\partial p_j} }{\frac{\partial R_j}{\partial p_j} - q_j} \right)}_{1 + \epsilon_{jj}^{-1}} .
\]
That is, the analyst can replace the term $D_{j \to k} \frac{p_k}{p_j}$ with the revenue diversion ratio $D_{j \to k}^R$, multiplied by an ``adjustment factor'' $(1 + \epsilon_{jj}^{-1})$. In the following sections, I show how to estimate revenue diversion ratios and own-price elasticities of demand based on revenue and margin data.

\paragraph*{Sum of Diversion Ratios Over Products} 
The summation property of diversion ratios depends on the underlying demand model. In pure discrete choice models, where total quantity is fixed, quantity diversion ratios from a product to all alternatives (including the outside option) sum to one \citep{conlon2021empirical}. However, revenue diversion ratios generally do not, since total market revenue may vary. Conversely, in discrete-continuous choice models with a fixed consumer budget, revenue diversion ratios sum to one, but quantity diversion ratios need not, as the total number of units consumed can vary. Note that in a standard pure discrete choice model, where the outside good is assumed to have a price of zero, the revenue diversion ratio to the outside option is necessarily zero.

%\todo{Is there a rigorous result that allows me to use uni-dimensional price?}

\section{Identification of Unilateral Effects \label{section:3.identification.of.unilateral.effects}}

In this section, I show how to identify first-order unilateral effects with the data on merging firms' revenues and margins.

\subsection{Assumptions on Data}
\begin{assumptionD}[Baseline data] \label{assumption:data}
    The analyst observes revenues and relative margins for the merging parties' products.
\end{assumptionD}
\begin{assumptionD}[Revenue diversion ratios data] \label{assumption:revenue.diversion.ratios.data}
The analyst observes the revenue diversion ratios for all pairs of the merging parties' products.
\end{assumptionD}
\begin{assumptionD}[Merger-specific efficiencies] \label{assumption:cost.savings}
The analyst observes the percentage decrease in marginal costs for the merging parties' products.
\end{assumptionD}
\begin{assumptionD}[Pass-through rates] \label{assumption:upp.approximates.true.price.effect} 
The analyst observes the merger pass-through rates.
\end{assumptionD}

Assumption \ref{assumption:data} requires product-level revenues and margins for the merging firms but does not require separate data on prices and quantities.\footnote{To enhance readability, I label assumptions about the data availability with the prefix D and those about model primitives with the prefix M.} Assumption \ref{assumption:revenue.diversion.ratios.data} departs from the standard focus on quantity diversion ratios by instead requiring revenue diversion ratios between the merging firms' products. Assumption \ref{assumption:cost.savings} requires the analyst to know the expected merger-specific marginal cost savings that can offset upward pricing incentives. Finally, Assumption \ref{assumption:upp.approximates.true.price.effect} is needed to translate upward pricing pressure into predicted price effects. While some of these assumptions may be demanding, or even more so than those underlying standard approaches based on price and quantity data, they provide a foundation for developing an alternative approach when traditional tools are inapplicable. I defer a more detailed discussion of the measurement challenges and practical implementation to Section \ref{section:unilateral.effects.discussion}.

\subsection{Identification of Gross Upward Pricing Pressure Index}

A primary measure of interest in this paper is the \emph{gross upward pricing pressure index (GUPPI)}, defined as the upward pricing pressure normalized by the pre-merger price to make the statistic unit-free, i.e., $\mathit{GUPPI}_j \equiv \mathit{UPP}_j / p_j$ \citep{salop2009updating, moresi2010use}. Consider a merger between two firms $A$ and $B$. As I derive in Online Appendix \ref{section:B.review.of.unilateral.effects.analysis}, the GUPPI associated with product $j \in \mathcal{J}_A$ can be written as
\begin{equation}\label{equation:gross.upward.pricing.pressure}
    \mathit{GUPPI}_j = \ddot{c}_j (1-m_j) + \sum_{k \in \mathcal{J}_B} m_k D_{j \to k} \frac{p_k}{p_j},
\end{equation}
where $\ddot{c}_j \equiv  (c_j^\text{post} - c_j^\text{pre})/ c_j^\text{pre} \in (-1,0)$ represents the percentage decrease in marginal cost from the pre-merger equilibrium due to merger-specific efficiencies; by convention, I omit the superscript ``pre'' when it is clear the object is being evaluated at the pre-merger equilibrium. Firm $B$'s GUPPIs are defined symmetrically.

I characterize GUPPIs as functions of data in Assumptions \ref{assumption:data}--\ref{assumption:cost.savings} as follows. Using Lemma \ref{lemma:relationship}.\ref{lemma:relationship.item.3}, I can rewrite the first-order conditions for an arbitrary firm $F \in \mathcal{F}$ (equation \eqref{equation:optimal.pricing.equation.2}) and the gross upward pricing pressure indices for product $j\in \mathcal{J}_A$ (equation \eqref{equation:gross.upward.pricing.pressure}) as
\begin{gather}
    -\epsilon_{jj}^{-1} - m_j + (1 + \epsilon_{jj}^{-1}) \sum_{l \in \mathcal{J}_F \backslash j} m_l D_{j \to l}^R = 0, \label{equation:first.order.conditions.rewritten} \\
    \mathit{GUPPI}_j = \ddot{c}_j ( 1- m_j) + (1 + \epsilon_{jj}^{-1}) \sum_{k \in \mathcal{J}_B} m_k D_{j \to k}^R, \label{equation:guppi.rewritten}
\end{gather}
respectively. Solving \eqref{equation:first.order.conditions.rewritten} for $\epsilon_{jj}$ characterizes the own-price elasticity of demand as a closed-form function of own-firm margins and revenue diversion ratios.
\begin{lemma}[Identification of own-price elasticities of demand] \label{lemma:own.price.elasticity}
Let $F \in \mathcal{F}$ be an arbitrary firm. For each product $j \in \mathcal{J}_F$,
    \begin{equation}\label{equation:own.price.elasticity}
    \epsilon_{jj} = - \frac{
    1 - \sum_{k \in \mathcal{J}_F \backslash j} m_k D_{j \to k}^R
    }{
    m_j - \sum_{k \in \mathcal{J}_F \backslash j} m_k D_{j \to k}^R
    }.
    \end{equation}
\end{lemma}
When firms produce only a single product, equation \eqref{equation:own.price.elasticity} reduces to the familiar single-product Lerner equation, $\varepsilon_{jj} = -1/m_j$.  While the standard demand estimation approach \`{a} la \citet{rosse1970estimating} recovers the markup $m_j$ from an estimate of the own-price elasticity $\varepsilon_{jj}$, I take the reverse approach. Viewed this way, Lemma \ref{lemma:own.price.elasticity} generalizes the idea to the multi-product setting, but clarifying that within-firm margins and revenue diversion ratios are required to identify the firm's own-price elasticities of demand.

Plugging in \eqref{equation:own.price.elasticity} into \eqref{equation:guppi.rewritten} expresses GUPPIs as known functions of merging firms' margins, revenue diversion ratios, and cost savings. Accordingly, my data assumptions allow for the identification of GUPPIs (but not unnormalized UPPs).
%\footnote{Once the own-price elasticities are identified, the cross-price elasticities can also be recovered using the relationships $D_{jk}^R = - \frac{\epsilon_{kj}^R}{\epsilon_{jj}^R} \frac{R_k}{R_j}$, $\epsilon_{jj}^R = \epsilon_{jj} + 1$, and $\epsilon_{kj}^R = \epsilon_{kj}$.}

\begin{proposition}[Identification of GUPPI] \label{proposition:guppis.are.identified}
Under Assumptions \ref{assumption:data},  \ref{assumption:revenue.diversion.ratios.data}, and \ref{assumption:cost.savings}, the gross upward pricing pressure indices are identified.
\end{proposition}

Unlike the standard approaches, Proposition \ref{proposition:guppis.are.identified} leverages the firms' profit maximization conditions to identify GUPPI. When prices, costs, and quantity diversion ratios are observed---as in \citet{farrell2010antitrust}---upward pricing pressures and thus GUPPIs can be directly calculated, leaving no role for the firms' first-order conditions. However, when prices are unobserved, the analyst can instead use the first-order conditions to recover the merging firms' own-price elasticities via Lemma \ref{lemma:own.price.elasticity}. Combined with margins and revenue diversion ratios, these elasticities identify the GUPPIs. 

When prices are unobserved, economists often approximate GUPPI by assuming a revenue proxy for quantities to compute quantity diversion ratios and that $p_k / p_j \approx 1$ \citep{ferguson2023economics}. However, as I demonstrate through my empirical application in Section \ref{section:6.empirical.application}, these approximations are not innocuous. My identification results show that such ad hoc assumptions are unnecessary.

\subsection{Identification of Merger Price Effects}
GUPPIs are informative of the signs of price effects, but translating GUPPIs to merger price effects requires pass-through rates \citep{farrell2010antitrust}. Based on a first-order approximation argument, \citet{jaffe2013first} shows merger price effects can be approximated as $\ddot{p} \approx M \cdot \mathit{GUPPI}$, where $\ddot{p} = (\ddot{p}_j)_{j \in \mathcal{J}_A \cup \mathcal{J}_B}$ is the vector of percentage change in prices due to merger, and $M$ is the \emph{merger pass-through matrix}.\footnote{Note that my merger pass-through rate matrix differs slightly from \citet{jaffe2013first}'s since I use GUPPIs instead of UPPs. Let $\Lambda$ be a diagonal matrix of $1/p_j$'s. \citet{jaffe2013first} shows $\Delta p  \approx M^* \cdot \mathit{UPP}$. Since $\mathit{GUPPI} = \Lambda \mathit{UPP}$ and $\ddot{p} = \Lambda \Delta p$, I get $\ddot{p} = \Lambda \Delta p = \Lambda M^* \mathit{UPP} = \Lambda M^* \Lambda^{-1} \Lambda \mathit{UPP} = \Lambda M^* \Lambda^{-1} \mathit{GUPPI}$, so $M = \Lambda M^* \Lambda^{-1}$.} Thus, once GUPPIs are known, the knowledge of $M$ identifies the merger price effects up to first-order approximation.

\begin{proposition}[UPP and merger price effects] \label{proposition:upp.and.merger.price.effects}
    Under Assumptions \ref{assumption:data},  \ref{assumption:revenue.diversion.ratios.data}, \ref{assumption:cost.savings}, and \ref{assumption:upp.approximates.true.price.effect}, the merger price effects in percentage terms are identified up to first-order approximation.
\end{proposition}

\subsection{Identification of Welfare Effects}
The first-order approach to merger analysis yields a vector of predicted price changes by the merging firms $(\ddot{p}_j)_{j \in \mathcal{J}_A \cup \mathcal{J}_B}$, holding the prices of non-merging firms fixed. To translate these predicted price effects into welfare implications without imposing a specific functional form for demand, I adopt the following first-order approximation approach.

First, I approximate the total consumer welfare effect as 
\begin{equation}\label{equation:total.welfare.effect}
    \Delta \mathit{CS} = \sum_{j \in \mathcal{J}_A \cup \mathcal{J}_B} \Delta \mathit{CS}_j,
\end{equation}
which sums the product-specific consumer surplus changes $\Delta \mathit{CS}_j$, each based on a \emph{ceteris paribus} price change of product $j$. Measure \eqref{equation:total.welfare.effect} ignores the cross-price effects but is popular since the additive separability allows for a simple calculation \citep{araar2019prices}.\footnote{When consumer preferences are non-homothetic, calculating the welfare impact of multiple price changes can be complicated because the value of the integral may depend on the sequence of price changes---a challenge known in the literature as the \emph{path-dependency problem} \citep{chipman1980compensating}. A common workaround is to approximate the total welfare effect as the sum of welfare changes from individual price changes, thereby avoiding this issue.} The impact on producer surplus is measured analogously.

Next, following \citet{jaffe2013first}, I approximate each consumer surplus variation due to a \emph{ceteris paribus} price increase of product $j$ as $\Delta \mathit{CS}_j \approx -\Delta p_j q_j^0$, where $\Delta p_j = (p_j^1 - p_j^0)>0$. The producer surplus variation (i.e., the impact on the profit of product $j$'s producer) is $\Delta \mathit{PS}_j = (\Delta p_j - \Delta c_j)q_j^1 + (p_j^0 - c_j^0)(q_j^1 - q_j^0)$, where $\Delta c_j = (c_j^1 - c_j^0) < 0$ accounts for the possible marginal cost savings. Based on the first-order approximation $q_j^1 \approx q_j^0 + (\frac{\partial q_j}{\partial p_j}\vert_{p_j=p_j^0})(p_j^1 - p_j^0) = q_j^0(1 + \epsilon_{jj} \ddot{p}_j)$, the following lemma expresses each $\Delta \mathit{CS}_j$ and $\Delta \mathit{PS}_j$ as functions of statistics that are known or identifiable under my previous data assumptions.

\begin{lemma}[Approximation of welfare effects]\label{lemma:approximation.of.welfare.effects}
The first-order approximations to the changes in consumer surplus and producer surplus associated with a price increase of product $j$ are $\Delta \mathit{CS}_j \approx - \ddot{p}_j R_j$ and $\Delta \mathit{PS}_j \approx (\ddot{p}_j - \ddot{c}_j (1-m_j)) R_j ( 1 + \epsilon_{jj} \ddot{p}_j) + \epsilon_{jj} R_j \ddot{p}_j m_j$, respectively.
\end{lemma}

Since the own-price elasticities can be identified using data on revenues, margins, and revenue diversion ratios, the same data suffices to identify changes in consumer and producer surplus, given estimates of the merger-induced price effects $(\ddot{p}_j)_{j \in \mathcal{J}_A \cup \mathcal{J}_B}$.

\begin{proposition}[Identification of welfare effects]\label{proposition:consumer.welfare}
Under Assumptions \ref{assumption:data}, \ref{assumption:revenue.diversion.ratios.data}, \ref{assumption:cost.savings}, and  \ref{assumption:upp.approximates.true.price.effect}, the welfare effects of a merger---measured in terms of consumer surplus and producer surplus---are identified up to a first-order approximation. 
\end{proposition}

Two remarks are in order. First, one can refine the estimate of consumer surplus loss by calculating the area of the trapezoid under a linearly approximated demand curve. This yields the approximation $\Delta \mathit{CS}_j^* \approx -(\Delta p_j q_j^0 + \frac{1}{2} \Delta p_j \Delta q_j) \approx - \ddot{p}_j R_j ( 1 + \frac{1}{2} \epsilon_{jj} \ddot{p}_j)$. Similarly, an upper bound on consumer surplus loss can be calculated as $\Delta \mathit{CS}_j^{**} \approx - \Delta p_j q_j^1 \approx - \ddot{p}_j R_j (1 + \epsilon_{jj} \ddot{p}_j)$. All three measures of consumer surplus loss from a ceteris paribus price change are identified under the same data assumptions and satisfy the ordering $\Delta \mathit{CS}_j < \Delta \mathit{CS}_j^* < \Delta \mathit{CS}_j^{**} < 0$.

Second, while the above relies on a first-order approximation to avoid imposing a specific functional form on demand, more precise welfare estimates are possible under a functional form assumption. For example, as shown in Online Appendix \ref{sec:identification.of.compensating.variation.under.ces.preferences}, compensating variation associated with a vector of price changes can be calculated exactly under a CES preference assumption. This parallels the standard consumer surplus calculation under logit demand, where simultaneous price changes are captured through changes in the inclusive value rather than treated additively.

\subsection{Identification of Compensating Marginal Cost Reductions}
I show how the \emph{compensating marginal cost reductions (CMCR)}---defined as the percentage decrease in the merging parties' costs that would leave the pre-merger prices unchanged after the merger \citep{werden1996robust}---are identified under the previous data assumptions. Let $(c_j^1)_{j \in \mathcal{J}_A \cup \mathcal{J}_B}$ be the vector of post-merger marginal costs that leave the merging parties' prices unchanged after the merger. Let $\ddot{c}_j = (c_j^1 - c_j^0)/c_j^0$ for $j \in \mathcal{J}_A \cup \mathcal{J}_B$ denote the compensating marginal cost reductions, represented in percentage terms relative to pre-merger marginal costs. The following lemma characterizes compensating marginal cost reductions as functions of pre-merger own-price elasticities, margins, and revenue diversion ratios.

\begin{lemma}[CMCR]\label{lemma:cmcr}
The compensating marginal cost reductions for each product $j \in \mathcal{J}_A \cup \mathcal{J}_B$ is equal to
\begin{equation}\label{equation:cmcr.calculation.from.margin}
    \ddot{c}_j = \frac{m_j^0 - m_j^1}{1 - m_j^0},
\end{equation}
where the post-merger margins $m_j^1$'s are implicitly defined by the post-merger first-order conditions
\begin{equation}\label{equation:post.merger.foc.for.cmcr} 
    - \epsilon_{jj}^{-1} - m_j^1 + (1 + \epsilon_{jj}^{-1})\sum_{l \in \mathcal{J}_A \backslash j} m_l^1 D_{j \to l}^R + (1 + \epsilon_{jj}^{-1}) \sum_{k \in \mathcal{J}_B } m_k^1 D_{j \to k}^R = 0.
\end{equation}
\end{lemma}

Thus, the analyst can calculate the CMCRs by first solving for post-merger margins $m^1 = (m_j^1)_{j \in \mathcal{J}_A \cup \mathcal{J}_B}$ implied by the system of linear equations in \eqref{equation:post.merger.foc.for.cmcr} (evaluating the own-price elasticities and diversion ratios at their pre-merger values) and plugging the solutions to \eqref{equation:cmcr.calculation.from.margin}.

\begin{proposition}[Identification of CMCR]\label{proposition:cmcr}
Under Assumptions \ref{assumption:data} and \ref{assumption:revenue.diversion.ratios.data}, the compensating marginal cost reductions are identified.
\end{proposition}

\citet{werden1996robust} shows that CMCRs are identified when prices, quantities, margins, and quantity diversion ratios are observed. In contrast, I show that CMCRs can be identified using revenues, margins, and revenue diversion ratios. This represents a weaker data requirement, since Werden’s assumptions imply mine, but not the reverse. The key distinction is that, although CMCRs are rooted in firms’ first-order conditions, Werden’s formula does not explicitly use them.

% \begin{remark}
% Setting $\mathit{GUPPI}_j = 0$ and solving for $\psi_j$ gives the critical value of cost-reduction for product $j$:
% \begin{equation}\label{equation:mc.reduction.to.offset.upp}
%     \psi_j^* = \left( 1 - m_j \right)^{-1} \left( \sum_{k \in \mathcal{J}_g} m_k D_{jk} \frac{p_k}{p_j} \right).
% \end{equation}
% The analyst can report \eqref{equation:mc.reduction.to.offset.upp} to gauge the reduction in marginal costs required to offset the upward pricing pressure. Note, however, that \eqref{equation:mc.reduction.to.offset.upp} differs from CMCR. Upward pricing pressure assumes other firms' marginal costs remain fixed. In contrast, CMCR assumes the costs decrease simultaneously and solves simultaneous equations.\footnote{\citet{farrell2010antitrust} endorses upward pricing pressure on grounds of simplicity and transparency while acknowledging that compensating marginal cost reductions can be more accurate.}
% \end{remark}
\subsection{Meeting Data Assumptions \label{section:unilateral.effects.discussion}}

My approach avoids the standard reliance on price and quantity data by instead using information on the merging firms' revenues, margins, and revenue diversion ratios.  However, measurement challenges are common in practice. I therefore briefly discuss how the data requirements  in Assumptions \ref{assumption:data}--\ref{assumption:upp.approximates.true.price.effect} can be addressed.

\paragraph{Assumption \ref{assumption:data}: Measuring Revenues and Margins} Antitrust authorities can often measure revenues and margins with reasonable accuracy during merger reviews by accessing firms' financial records \citep{farrell2010antitrust}. When this approach is infeasible, researchers may turn to alternatives such as using information from public filings (see, e.g., \citet{smith2004supermarket, ellickson2020measuring}) or estimating markups through production function or cost-share methods (see, e.g., \citet{de2011product, loecker2012markups, gandhi2020identification, de2022hitchhiker, raval2023testing, kirov2023measuring}). When product-level margins are difficult to estimate, firm-level margins can serve as a useful proxy, especially under CES preferences, which imply firms set uniform relative margins across their products \citep{nocke2018multiproduct}. 

\paragraph*{Assumption \ref{assumption:revenue.diversion.ratios.data}: Measuring Revenue Diversion Ratios} 

A key distinction of my framework is that it relies on the availability of revenue diversion ratios, rather than the price and quantity data used in standard approaches. The approaches for measuring revenue diversion ratios are methodologically similar to those used for quantity diversion ratios but rely on different data. I outline several approaches as follows. First, as noted in \citet{farrell2010antitrust}, firms often track diversions in the ordinary course of business, such as identifying competitors they lose or gain customers from. Antitrust authorities can use these internal documents and financial records to infer diversion patterns. While \citet{farrell2010antitrust}'s framework centers on quantity diversion ratios, their statement regarding the accessibility of diversion data is equally applicable to revenue diversion ratios, as firms commonly track their performance in terms of revenue.

Second, analysts can estimate revenue diversion ratios using econometric analysis of quasi-experiments---such as product entries or exits---that reveal how revenues shift across products \citep{conlon2021empirical}. For example, in \emph{FTC v. Kroger}, the FTC's economic expert, Dr. Nicholas Hill, used a 2022 strike at King Soopers (a Kroger banner) in Denver, which boosted sales at nearby Safeway (an Albertsons banner) stores, to supplement his market share-based diversion estimates between the merging parties \citep{ftc2024kroger}.

Third, revenue diversion ratios can be estimated using consumer surveys that ask how spending would shift under hypothetical price or product changes---a method frequently used by the UK CMA \citep{cma2011good, cma2017retail}. While surveys may lower the cost of measuring quantity diversion ratios, revenue-based measures remain valuable for two reasons. First, they better capture substitution when quantities are ambiguous---for example, if shoppers change their basket composition rather than store visits. Second, quantity ratios alone are insufficient for identifying unilateral effects when prices are unobserved, as GUPPI calculations require relative prices.

Finally, revenue diversion ratios can be estimated from cross-sectional consumer expenditure data using a discrete-continuous demand framework. I demonstrate this under CES preferences in Section \ref{section:4.consumer.demand} and extend it to a broader class of models in Appendix \ref{subsection:revenue.diversion.ratios.under.discrete.continuous.demand} assuming homogeneous price responsiveness. While “price diversion ratios” (from small price changes) and “forced diversion ratios” (from product removals) may differ, they tend to align when consumer price sensitivity is relatively uniform \citep{conlon2021empirical}. In fact, under CES demand—as in the logit case—I show that the two coincide.

\paragraph{Assumption \ref{assumption:cost.savings}: Measuring Cost Efficiencies} Merger-specific cost efficiencies are challenging to predict, as they depend on forward-looking claims that are often unverifiable and may reflect overly optimistic assumptions by the merging firms. Antitrust authorities typically draw on firms' financial records or credit some default values of efficiencies, e.g., assuming a 5\% reduction in marginal costs for each product \citep{farrell2010antitrust}. A potential alternative is to estimate firms' production functions to predict merger-specific efficiencies \citep{grieco2018brewed, khmelnitskaya2024identifying}, but applying this approach without price and quantity data remains an open challenge.

\paragraph{Assumption \ref{assumption:upp.approximates.true.price.effect}: Measuring Merger Pass-through Rates} The literature has long recognized the difficulty of estimating pass-through rates, as they depend on demand curvature \citep{farrell2010antitrust, jaffe2013first, mackay2014bias, miller2016pass, miller2017pass, miller2021quantitative}. Nonetheless, when reliable pass-through data are unavailable and demand estimation is infeasible, researchers have proposed two practical strategies to address this challenge. One approach is to impose a functional form on demand, which allows pass-through rates to be expressed as known functions of observables; I adopt this approach in my first empirical application in Section \ref{section:6.empirical.application} by assuming CES demand.\footnote{In Online Appendix \ref{section:pass.through.matrix.under.ces.preference}, I express merger pass-through matrix as known functions of merging firms' revenues and margins under a CES preference assumption.} Alternatively, one can simply assume $M \approx I$, which implies $\ddot{p}_j \approx \mathit{GUPPI}_j$ \citep{miller2017upward, dutra2020antitrust}; I use this strategy in my second empirical application in Section \ref{section:7.empirical.application.II}.\footnote{In Online Appendix \ref{section:accuracy.of.upward.pricing.pressure}, I provide simulation evidence that upward pricing pressure is a conservative predictor of the true merger price effect when the demand is CES, which complements \citet{miller2017upward}'s findings.}

% Our goal is to relax data requirements to identify gross upward pricing pressure indices and welfare effects from a merger. Our approach is reminiscent of the sufficient statistics approach because we characterize objects that are sufficient to identify GUPPIs and welfare effects, which need not be estimated from structural models. We also suggest reduced-form approaches or how ordinary course documents may be used to identify the welfare effects.
\section{Consumer Demand and Revenue Diversion Ratios \label{section:4.consumer.demand}}

This section shows that assuming CES demand \citep{spence1976product, dixit1977monopolistic} facilitates the identification of revenue diversion ratios using only consumer expenditure data for the merging firms. Appendix \ref{subsection:revenue.diversion.ratios.under.discrete.continuous.demand} extends the identification argument to a broader class of discrete-continuous demand models.

\subsection{Assumptions on Consumer Preference and Data}

I specify consumer-side model primitives as a tuple $\langle \mathcal{I}, (\mathcal{C}_i, U_i, B_i)_{i \in \mathcal{I}} \rangle$, where $\mathcal{I}$ is the set of consumers, $\mathcal{C}_i$ is consumer $i$'s consideration set, $U_i: \mathbb{R}^{\dim (\mathcal{C}_i)} \to \mathbb{R}$ is consumer $i$'s utility function, and $B_i \in \mathbb{R}_+$ is $i$'s budget.\footnote{While consumers’ consideration sets can be omitted from the model primitives without loss of generality, making them explicit is conceptually and computationally useful. In my grocery merger application in Section \ref{section:7.empirical.application.II}, consideration sets are determined by the distance between consumers’ residences and grocery stores.} I assume consideration sets always include an outside option $j = 0$. To simplify the exposition, I assume that the analyst knows every consumer's consideration set $\mathcal{C}_i$ and budget $B_i$, although the analyst may need to estimate them in practice.

Each consumer $i \in \mathcal{I}$ maximizes utility $U_i$ with respect to consumption vector $q_i \equiv (q_{ij})_{j \in \mathcal{C}_i}$ subject to a budget constraint $\sum_{j \in \mathcal{C}_i} p_j q_{ij} \leq B_i$. The Marshallian demand function induces consumer $i$'s expenditure on each product $j$ as $e_{ij} = \alpha_{ij}B_i$, where $\alpha_{ij} \in [0,1]$ represents the expenditure share of the budget allocated to product $j$. Product $j$'s revenue is the sum of consumers' expenditure on the product, $R_j = \int_{i \in \mathcal{S}_j} e_{ij} d\mu$, where $i \in \mathcal{S}_j$ ($i$ is a shopper of $j$) if and only if $j \in \mathcal{C}_i$ ($i$ considers $j$), and $\mu$ is a measure on the set of consumers.\footnote{For instance, \citet{ellickson2020measuring} estimate a spatial demand model of grocery competition, where each $i \in \mathcal{I}$ is a census tract and each $j \in \mathcal{J}$ is a grocery store. Consumers consider all stores within 15 miles, defining $\mathcal{C}_i$. Grocery budgets are set at 13\% of tract income. With a counting measure $\mu$, total revenue of store $j$ simplifies to $R_j = \sum_{i \in \mathcal{S}_j} e_{ij}$.}

I maintain the following assumptions on data and consumers' preferences.

\begin{assumptionD}[Consumer expenditure share data]\label{assumption:budget.and.consumer.expenditure.data} The analyst observes consumers' expenditure shares on the merging firms' products.
\end{assumptionD}

\begin{assumptionM}[CES utility]\label{assumption:ces.utility}
Consumers' preferences can be described by a constant elasticity of substitution (CES) function $U_i(q_i) = A \left( \sum_{j \in \mathcal{C}_i } \beta_{ij}^{\frac{1}{\eta}}q_{ij}^{\frac{\eta-1}{\eta}}  \right)^{\frac{\eta}{\eta-1}}$ with parameters $A, \eta, \beta_{ij}>0$.\footnote{Parameter $A$ is an arbitrary scaling factor (which I do not use), $\eta$ is the elasticity of substitution between products, and $\beta_{ij}$ is the product $j$'s quality perceived by consumer $i$. Note that (i) if $\eta \to \infty$, $U_i (q_i) \to A \sum_{j \in \mathcal{C}_i} q_{ij}$; (ii) if $\eta \to 1$, $U_i(q_i) \to (A /\prod_{j \in \mathcal{C}_i} \beta_{ij}^{\beta_{ij}}) \prod_{j \in \mathcal{C}_i} q_{ij}^{\beta_{ij}}$; and (iii) if $\eta \to 0$, $U_i(q_i) \to A \min_{j \in \mathcal{C}_i} \{q_{ij}/\beta_{ij}\}$. }
\end{assumptionM}

Assumption \ref{assumption:budget.and.consumer.expenditure.data} outlines the data requirements for the forthcoming identification result and, importantly, does not require information on consumer spending for non-merging firms. It can be satisfied using the merging parties' transaction data, consumer surveys on budget allocation, or a parametric estimation of consumer expenditure functions \citep{hosken2016horizontal, reynolds2008use, ellickson2020measuring}.\footnote{For instance, in the recent Kroger/Albertsons merger, the FTC's expert used loyalty card data from both firms to nonparametrically estimate consumer expenditure shares at the census block-group level. The defendants' expert instead relied on a parametric model developed by \citet{ellickson2020measuring} to estimate such shares.} Assumption \ref{assumption:ces.utility} formalizes the CES demand structure, where $\beta_{ij}$ captures the appeal (or quality) of product $j$ to consumer $i$, and $\eta$ denotes the elasticity of substitution, which also reflects how sensitive consumers are to prices. 

\subsection{Identification of Revenue Diversion Ratios Under CES Preference}

\paragraph{Identification with Consumer Expenditure Data}
Assumption \ref{assumption:ces.utility} yields the following characterizations of expenditure shares, own-price elasticities, and revenue diversion ratios, which will allow me to leverage Assumption \ref{assumption:budget.and.consumer.expenditure.data}.

\begin{lemma}[Characterization under CES preference] \label{lemma:revenue.diversion.ratio.under.ces.utility} Under Assumption \ref{assumption:ces.utility}, consumers' expenditure shares, own-price elasticities of revenue, and revenue diversion ratios ($j \neq k$) are
\begin{align}
    \alpha_{ij} & = \frac{\exp(u_{ij})}{\sum_{k \in \mathcal{C}_i} \exp(u_{ik})}, &&\text{where } u_{ij} = \log \beta_{ij} + (1-\eta) \log p_j; \label{equation:ces.expenditure.share} \\
    \epsilon_{jj}^R & = \int_{i \in \mathcal{S}_j} \overline{w}_{ij} \epsilon_{i, jj}^e d\mu,  &&\text{where } \overline{w}_{ij} = \frac{\alpha_{ij} B_i}{\int_{\tilde{i} \in \mathcal{S}_j} \alpha_{\tilde{i}j} B_{\tilde{i}} d\mu}, \; \epsilon_{jj}^e = (1-\eta) (1-\alpha_{ij}) ; \label{equation:ces.own.price.elasticity.of.revenue} \\
    D_{j \to k}^R & = \int_{i \in \mathcal{S}_k} w_{ij} D_{i, j \to k}^e d\mu, &&\text{where } w_{ij} = \frac{\alpha_{ij} (1- \alpha_{ij})B_i}{\int_{\tilde{i} \in \mathcal{S}_j} \alpha_{\tilde{i}j}(1 - \alpha_{\tilde{i}j}) B_{\tilde{i}} d\mu}, D_{i, j\to k}^e =  \frac{\alpha_{ik}}{1 - \alpha_{ij}}. \label{equation:ces.revenue.diversion.ratio}
\end{align}
\end{lemma}

Lemma \ref{lemma:revenue.diversion.ratio.under.ces.utility} establishes key characterizations of CES demand that simplify identification. First, each consumer's expenditure shares can be written as softmax functions of utility indices $(u_{ij})_{j \in \mathcal{C}_i}$, which is reminiscent of the choice probability formula under logit demand. Next, the own-price elasticities of revenue are weighed sums of individual-level own-price elasticities of expenditure $\epsilon_{i,jj}^e \equiv \frac{\partial e_{ij}}{\partial p_j} \frac{p_j}{e_{ij}}$, which in turn depends on the value of $\eta$. Similarly, revenue diversion ratios can be expressed as weighted sums of individual-level expenditure diversion ratios $D_{i,j \to k}^e \equiv - \frac{\partial e_{ik} / \partial p_j}{\partial e_{ij} / \partial p_j}$.\footnote{A logit-demand analog appears in \citet{hosken2016horizontal}. If there is a single representative consumer in the market, then the diversion ratio formula collapses to $D_{j \to k}^R = \frac{\alpha_k}{1 - \alpha_j}$, so the merging firms' product-specific market shares (in revenues) are sufficient to identify their revenue diversion ratios. } In both cases, the weights are known functions of consumers' expenditure shares. Lemma \ref{lemma:revenue.diversion.ratio.under.ces.utility} is consistent with the literature's finding that CES and logit demands share similar simplification properties as well as limitations (see, e.g., \citet{anderson1989demand, nocke2018multiproduct, dube2022discrete}).\footnote{Like logit, CES imposes restrictive substitution patterns, as cross-price elasticities depend only on the product receiving the price change \citep{nevo2011empirical, berry2021foundations, conlon2021empirical}. This makes CES less suitable for capturing rich substitution patterns. In my analysis of the Albertsons-Safeway merger (Section \ref{section:7.empirical.application.II}), I use nested CES preferences to balance tractability and flexibility.}

The diversion ratio formula \eqref{equation:ces.revenue.diversion.ratio} implies that revenue diversion ratios can be identified from consumer expenditure share data.
\begin{proposition}[Identification under CES preference] \label{proposition:identification.of.diversion.ratios.under.ces.utility}
Suppose consumers' preferences satisfy Assumption \ref{assumption:ces.utility}. Then, if Assumption \ref{assumption:budget.and.consumer.expenditure.data} holds, the revenue diversion ratios between all merging firms' products are identified.
\end{proposition}

\paragraph{Identification with Second-Choice Expenditure Data} 

An alternative approach to identifying diversion ratios is to rely on their mathematical definition and observe how consumers respond to (quasi-)experimental price changes. However, small price variations necessary for identification may be unavailable. In such cases, the following property of CES demand justifies using second-choice expenditure data or intertemporal variation in revenue following the removal of a product.

\begin{proposition}[Identification with second-choice data] \label{proposition:identification.with.second.choice.data} Under Assumption \ref{assumption:ces.utility}, the revenue diversion ratio from product $j$ to product $k$ following a removal of product $j$ is equal to the revenue diversion ratio derived from a marginal increase in price $j$'s price.
\end{proposition}

Proposition \ref{proposition:identification.with.second.choice.data} parallels the result in \citet{conlon2021empirical}, who show that quantity diversion ratios can be inferred from second-choice information under logit demand. While their notation of second-choice data reflects changes in discrete choice probabilities after a product's removal, the second-choice data in the present context capture how consumers reallocate their expenditures in response to a product's removal.

\section{Merger Simulation \label{section:5.merger.simulation}}
This section shows that a merger simulation is feasible without price data under CES preferences. Specifically, I show that although it is not possible to predict the post-merger prices $p_j^\text{post}$ or absolute price changes $\Delta p_j = p_j^\text{post} - p_j^\text{pre}$, it is possible to predict the percentage change relative to the pre-merger equilibrium $\ddot{p}_j = (p_j^\text{post} - p_j^\text{pre})/p_j^\text{pre}$. Thus, the proposed methodology allows the analyst to make statements such as ``In the post-merger equilibrium, the prices of products 1, 2, and 3 will be 3\%, 6\%, and 10\% higher, respectively, relative to the pre-merger equilibrium.'' I extend the argument to a class of discrete-continuous demand in Appendix \ref{subsection:merger.simulation.under.discrete.continuous.demand}. 

\begin{assumptionD}[Merger simulation] \label{assumption:merger.simulation.under.ces}
For all consumer $i \in \mathcal{I}$ and product $j \in \mathcal{J}$, the analyst observes the expenditure share $(\alpha_{ij})_{i \in \mathcal{I}, j \in \mathcal{J}}$, margins $(m_j)_{j \in \mathcal{J}}$, and percentage reductions in marginal costs $(\ddot{c}_j)_{j \in \mathcal{J}}$.
\end{assumptionD}

Assumption \ref{assumption:merger.simulation.under.ces} describes the data required for merger simulation. It is considerably more demanding than earlier assumptions, as it requires the analyst to observe expenditure shares, margins, and potential cost efficiencies for \emph{all} products in the market, including those of non-merging firms. Though stringent, these requirements clarify the distinction between merger simulations and the intermediate steps (e.g., demand estimation) that may precede them.\footnote{In practice, intermediate steps and assumptions are likely necessary to satisfy Assumption \ref{assumption:merger.simulation.under.ces}. A practical approach to satisfying Assumption \ref{assumption:merger.simulation.under.ces} is to estimate expenditure shares by combining a parametric utility specification \citep{ellickson2020measuring}, firm-level margin data \citep{nocke2018multiproduct}, and assumed efficiency credits \citep{farrell2010antitrust}.}

To build intuition for how Assumption \ref{assumption:merger.simulation.under.ces} enables merger simulation without price data, consider the post-merger profit-maximization conditions. As shown in equation \eqref{equation:first.order.conditions.rewritten}, these conditions can be written in terms of own-price elasticities, margins, and revenue diversion ratios. If the post-merger values of these objects can be expressed as known functions of the percentage price changes $\ddot{p}_j$ and observed pre-merger data, the first-order conditions can be written as a system $f(\ddot{p})=0$. The analyst can then solve for the equilibrium vector of relative price changes $(\ddot{p}_j^\text{*})_{j \in \mathcal{J}}$. The following lemma shows how the post-merger terms relate to an arbitrary vector $(\ddot{p}_j)_{j \in \mathcal{J}}$.

\begin{lemma}[Merger simulation under CES preferences] \label{lemma:merger.simulation.under.ces}
    Suppose Assumption \ref{assumption:ces.utility} holds. For each $j \in \mathcal{J}$, let $\ddot{p}_j = (p_j^\text{post} - p_j^\text{pre})/p_j^\text{pre}$, where $p_j$ is an arbitrary (possibly non-equilibrium) candidate of post-merger price. Then, the corresponding utility indices, consumer expenditure shares, own-price elasticities of demand, revenue diversion ratios, and margins induced by $\ddot{p} = (\ddot{p}_j)_{j \in \mathcal{J}}$ can be expressed as 
    \begin{align*}
        u_{ij}^\text{post} & = u_{ij}^\text{pre} + (1 - \eta) \log(1 + \ddot{p}_j), \\
        \alpha_{ij}^\text{post} & = \frac{\exp(u_{ij}^\text{post})}{\sum_{l \in \mathcal{C}_i} \exp (u_{il}^\text{post})}, \\
        \epsilon_{jj}^\text{post} & = (1 - \eta) \int_{i \in \mathcal{S}_j} \overline{w}_{ij}^\text{post} (1- \alpha_{ij}^\text{post}) d\mu - 1,  \\
        D_{j \to k}^{R,\text{post}} & = \int_{i \in \mathcal{S}_k} w_{ij}^\text{post} D_{i,j\to k}^{e, \text{post}} d\mu, \\
        m_j^\text{post} & = 1 - (1 - m_j^\text{pre}) \left( \frac{1 + \ddot{c}_j}{1 + \ddot{p}_j} \right),
    \end{align*}
    where $\overline{w}_{ij}^\text{post} = \frac{\alpha_{ij}^\text{post} B_i }{\int_{\tilde{i} \in \mathcal{S}_j} \alpha_{\tilde{i}j}^\text{post} B_{\tilde{i}} d\mu}$, $w_{ij}^\text{post} = \frac{\alpha_{ij}^\text{post}(1-\alpha_{ij}^\text{post}) B_i}{\int_{\tilde{i} \in \mathcal{S}_j} \alpha_{\tilde{i}j}^\text{post}(1-\alpha_{\tilde{i}j}^\text{post})B_{\tilde{i}} d\mu }$, and $D_{i,j \to k}^{e,\text{post}} = \frac{\alpha_{ik}^\text{post}}{1 - \alpha_{ij}^\text{post}}$.
\end{lemma}

Lemma \ref{lemma:merger.simulation.under.ces} shows how the relative price changes $(\ddot{p}_j)_{j \in \mathcal{J}}$ translate into the components of the post-merger first-order conditions. To complete the identification, the analyst must recover the pre-merger utility indices $(u_{ij}^\text{pre})_{i \in \mathcal{I}, j \in \mathcal{J}}$ and the elasticity of substitution parameter $\eta$. The pre-merger utility indices can be retrieved by inverting the softmax expression in equation \eqref{equation:ces.expenditure.share}, while $\eta$ can be (over-)identified using the own-price elasticity formula in equation \eqref{equation:ces.own.price.elasticity.of.revenue}, as elasticities are identified when margins and revenue diversion ratios are known (Lemma \ref{lemma:own.price.elasticity}). Under Assumption \ref{assumption:merger.simulation.under.ces}, then, the post-merger first-order conditions are fully determined by $\ddot{p}$, and the predicted relative price changes can be obtained by solving the resulting system of equations.

\begin{proposition}[Merger simulation under CES preferences] \label{proposition:merger.simulation.under.ces} Under Assumptions \ref{assumption:ces.utility} and \ref{assumption:merger.simulation.under.ces}, the percentage price changes from the pre-merger equilibrium to the post-merger equilibrium for all products in the market are identified.
\end{proposition}

\section{Empirical Application I: Staples/Office Depot (2016) \label{section:6.empirical.application}}
I apply my framework to evaluate the proposed merger of Staples and Office Depot, which was eventually blocked in 2016. I use this empirical example to illustrate the simplicity of my approach.

\subsection{Background}
Staples and Office Depot are the two largest suppliers of office supplies and services in the United States. On February 4, 2015, Staples entered into a \$6.3 billion merger agreement with Office Depot. On December 9, 2015, the Federal Trade Commission (FTC) sued to block the merger over concerns that it would significantly reduce competition in the national market for consumable office supplies sold to ``large business customers,'' defined as businesses that purchased at least \$500,000 worth of consumable office supplies during 2014. The FTC defined the market as a cluster of non-substitutable office supplies (e.g., pens, folders, Post-it notes)---which the court accepted---citing similarity in competitive conditions across products and analytic simplicity \citep{shapiro2016testimony}. On May 10, 2016, a federal judge granted the FTC's preliminary injunction to block the merger \citep{dccourt2016}. The parties subsequently announced that they would abandon the deal.

\subsection{Empirical Framework}
To focus on illustrating my methodology, I make several simplifying assumptions. First, I assume that the competition in the relevant market can be approximated by the standard Bertrand-Nash framework.\footnote{In practice, the U.S. market for consumable office supplies to large B-to-B customers involves a mix of auctions, bargaining, and flat discounts. Core products are typically procured through requests for proposals, while non-core items are often purchased at flat discounts off list prices. In this analysis, however, I abstract from these complexities and assume firms compete by posting prices.} Second, I assume that Staples and Office Depot are the main competitors in the market and group the rest of the (fringe) competitors as the outside option.\footnote{This is consistent with Veritiv, the third-largest firm in the FTC's proposed market, holding only a 5.2\% share, well behind Staples and Office Depot \citep{shapiro2016testimony}. Although my framework does not require grouping fringe firms as an outside option, doing so simplifies the merger simulation by avoiding the need for margin data on fringe firms.} Third, I assume that each firm is a seller of a single product, which is in line with the FTC's conceptualization of the relevant market as those offering a ``cluster'' of office products. Fourth, I assume zero merger-specific marginal cost savings because they are difficult to estimate. Finally, I assume a single representative consumer with CES preference. The above assumptions facilitate the evaluation of unilateral effects via revenue diversion ratios using data on total market size and the merging firms' revenues and margins.
\subsection{Data}
Since I cannot access the confidential data used in the investigation, I use publicly available documents to infer the parties' revenues, margins, and market shares. Based on the presentation material prepared by the FTC's expert Carl Shapiro, I infer the total market size in 2014 to be $B=\$2.05$ billion. The market shares of Staples and Office Depot (in terms of revenue) were $\alpha_{\text{SP}} = 47.3\%$ and $\alpha_{\text{OD}} = 31.6\%$, respectively \citep{shapiro2016testimony}. Based on the companies' 10-K documents for fiscal year 2014, I infer Staples' and Office Depot's margins to be $m_{\text{SP}} = 25.8\%$ and $m_{\text{OD}} = 23.4\%$, respectively.\footnote{Margins in the market for consumable office supplies to large business customers may be quite different from the overall margin reported in 10-K due to a variety of factors specific to the market, e.g., discounts and rebates. } 
%According to the companies' 10-K documents, Staples' and Office Depot's 2014 net revenues were \$5.6 billion and \$4.7 billion in the business-to-business channel.
\subsection{Results}

\paragraph*{Gross Upward Pricing Pressure Indices and Consumer Harm}
I estimate annual consumer harm as follows. From equation \eqref{equation:own.price.elasticity}, own-price elasticities implied by firms' margins (with $\epsilon_{jj} = -1/m_j$) are $\epsilon_{\text{SP},\text{SP}} = - 3.875$ and $\epsilon_{\text{OD}, \text{OD}} = - 4.273$. Equation \eqref{equation:ces.revenue.diversion.ratio} implies that the revenue diversion ratios under the single-agent CES assumption is $D_{j \to k}^R = \frac{\alpha_k}{1 - \alpha_j}$, so the revenue diversion ratios between the two firms are $D_{\text{SP} \to \text{OD}}^R = 59.9\%$ and $D_{\text{OD} \to \text{SP}}^R = 69.1\%$. The gross upward pricing pressure index formula \eqref{equation:guppi.rewritten} assuming zero merger-specific efficiency gives $\mathit{GUPPI}_j = (1 + \epsilon_{jj}^{-1}) m_k D_{j \to k}^R$, so I get $\mathit{GUPPI}_{\text{SP}} = 10.4\%$ and $\mathit{GUPPI}_{\text{OD}}= 13.7\%$. In Online Appendix \ref{section:pass.through.matrix.under.ces.preference}, I show that I can calculate the merger pass-through matrix using data on merging firms' revenues and margins under the CES preference assumption; the estimated merger pass-through matrix is $M= \left[ \begin{smallmatrix} 1.005 & 0.345 \\ 0.347 & 1.098 \end{smallmatrix} \right]$. The first-order merger price effects, calculated as $\ddot{p} \approx M \cdot \mathit{GUPPI}$, are $\ddot{p}_{\text{SP}} = 15.2\%$ and $\ddot{p}_{\text{OD}} = 18.7\%$.\footnote{Using GUPPI to approximate merger price effects (i.e., $\ddot{p}_j \approx \mathit{GUPPI}_j$) \`{a} la \citet{miller2017upward} would produce conservative estimates of merger price effects.} The estimated annual consumer harm, calculated using $\Delta \mathit{CS} = \Delta \mathit{CS}_\text{SP} + \Delta \mathit{CS}_\text{OD}$ with $\Delta \mathit{CS}_j \approx -\ddot{p}_j \alpha_j B$ (Lemma \ref{lemma:approximation.of.welfare.effects}), is $\$268.2$ million. 

\paragraph*{Compensating Marginal Cost Reductions}
Holding prices fixed at the pre-merger level, the post-merger margins required to generate marginal cost reductions necessary to offset upward pricing incentives can be obtained via equation \eqref{equation:post.merger.foc.for.cmcr}; solving the system of equations for post-merger margins gives $m_{\text{SP}}^1 = 47.3\%$ and $m_{\text{OD}}^1 = 48.5\%$. By plugging these estimates into equation \eqref{equation:cmcr.calculation.from.margin}, I obtain estimates of compensating marginal cost reductions as $\mathit{CMCR}_\text{SP} = 29.1\%$ and $\mathit{CMCR}_\text{OD} = 32.7\%$.

\paragraph*{Merger Simulation}
Suppose that Staples and Office Depot are the only players in the market and that all other products are included in the outside option. The mean utilities of a representative consumer that rationalize the pre-merger market shares of Staples and Office Depot can be computed by the softmax inversion $\log \alpha_k - \log \alpha_0 = u_k$ for $k=1,2$, which gives $u_{\text{SP}}^\text{pre} = 0.807$ and $u_{\text{OD}}^\text{pre} = 0.404$. Using equation \eqref{equation:ces.own.price.elasticity.of.revenue}, I estimate the elasticity of the substitution parameter to be $\eta = 6.121$.\footnote{Since $\epsilon_{jj}^R = (1-\alpha_j)(1-\eta)$ and $\epsilon_{jj} = \epsilon_{jj}^R - 1$, I can use $\eta = 1 - \frac{1 + \epsilon_{jj}}{1 - \alpha_j}$, which gives two equations to (over-)identify $\eta$. I calculate $\eta_{\text{SP}} = 6.457$ and $\eta_{\text{OD}} = 5.786$ using each equation. I take their average to arrive at $\eta = 6.121$.} Using Lemma \ref{lemma:merger.simulation.under.ces}, I set up the post-merger first-order condition as known functions of relative price changes to get $f(\ddot{p})=0$. Solving the system of equations for $\ddot{p}$ gives $\ddot{p}_{\text{SP}}^* = 14.3\%$ and $\ddot{p}_{\text{OD}}^* = 18.0\%$, which turns out to be close to those of the first-order approach. The corresponding annual consumer harm, again estimated with the expression in Lemma \ref{lemma:approximation.of.welfare.effects} for comparison, is $\ddot{p}_\text{SP}^* \alpha_\text{SP} B + \ddot{p}_\text{OD}^* \alpha_\text{OD} B = \$255.7$ million.

\paragraph{Conclusion} Analyses using gross upward pricing pressure indices, compensating marginal cost reductions, and merger simulation all lead to the same conclusion: the merger between Staples and Office Depot would result in substantial anticompetitive harm, supporting the FTC’s decision to challenge it. In this case, the unilateral effects measures are straightforward to compute and require only basic calculations.

\subsection{Comparison to Standard Approach}

As shown above, my approach allows for straightforward estimation of unilateral effects statistics using revenue diversion ratios derived from data on revenues and margins—information typically accessible to antitrust authorities through firms’ ordinary-course financial records (as was the case for the FTC in this matter). By contrast, estimating quantity diversion ratios would have required constructing price and quantity indices, a more complex process involving additional data and a series of intermediate assumptions. Even applying the proportional-to-share formula under the standard logit demand model poses challenges, as it depends on knowing total market size in quantity terms, which is arguably more difficult to obtain than market size in dollar terms.\footnote{The FTC's economic expert, Carl Shapiro, did not construct price or quantity indices but relied on win-loss data---tracking whether a B-to-B customer chose Staples when Office Depot lost a bid, or vice versa---to approximate the diversion ratio between the two firms as inputs to his hypothetical monopolist test.} 

When quantity data are unavailable, analysts often approximate quantity diversion ratios by applying the logit demand model but substituting revenue shares for quantity shares in the proportional-to-share formula. However, this can introduce substantial bias. Consider GUPPI as an example. A naive approach without price or quantity data would approximate $\mathit{GUPPI}_j = m_k D_{j \to k} \frac{p_k}{p_j} \approx m_k D_{j \to k}^R$, arguing revenue proxies for quantity and $p_k/p_j \approx 1$. This yields $\mathit{GUPPI}_\text{SP}^\text{naive} = 14.0\%$ and $\mathit{GUPPI}_\text{OD}^\text{naive} = 17.8\%$, both noticeably higher than the true values of $\mathit{GUPPI}_\text{SP} = 10.4\%$ and $\mathit{GUPPI}_\text{OD} = 13.7\%$, respectively. The upward bias stems from omitting the adjustment factor $(1 + \epsilon_{jj}^{-1})$ in the correct formula $\mathit{GUPPI}_j = (1 + \epsilon_{jj}^{-1})m_k D_{j \to k}^R$.

Next, consider CMCR. For a merger between single-product firms $j$ and $k$, the CMCR formula is $\mathit{CMCR}_j = \frac{m_j D_{j \to k} D_{k \to j} + m_k D_{j \to k} \frac{p_k}{p_j}}{(1 - m_j)(1 - D_{j \to k} D_{k \to j})}$ \citep{werden1996robust}. As before, substituting estimated revenue diversion ratios for quantity diversion ratios and assuming $p_k / p_j \approx 1$ lead to inflated estimates: $\mathit{CMCR}_\text{SP}^\text{naive} = 56.9\% $ and $\mathit{CMCR}_{\text{OD}}^\text{naive} = 61.4\%$, compared to the true values of $\mathit{CMCR}_\text{SP} = 29.1\%$ and $\mathit{CMCR}_\text{OD} = 32.7\%$, respectively.
\section{Empirical Application II: Albertsons/Safeway (2015) \label{section:7.empirical.application.II}}

As a second empirical application, I examine the 2015 merger between Albertsons and Safeway, which the FTC approved subject to the divestiture of 168 stores. Analyzing grocery mergers is often challenging due to limited access to price and quantity data and the difficulty of adapting them to fit standard tools. In contrast, store-level revenues and margins are typically available to antitrust authorities, as firms track them routinely in the course of standard accounting and management practices. I use the \emph{Albertsons/Safeway} merger to illustrate how my framework can be used to evaluate a large-scale retail merger in a highly tractable way, relying on a first-order approach that uses revenue and margin data, which are typically available to antitrust authorities.\footnote{I do not run a merger simulation because solving the first-order conditions for thousands of stores is computationally demanding.}\footnote{The same approach can be readily applied to the recent Kroger/Albertsons merger, which represents the largest proposed supermarket merger in U.S. history.} 

\subsection{Background}
In March 2014, AB Acquisition LLC, the parent of the Albertsons supermarket chain, agreed to terms to purchase Safeway for \$9 billion. Albertsons operated 1,075 stores in 28 states under the banners Albertsons, United, Amigos, and Market Street, among others. Safeway owned 1,332 stores in 18 states under the banners Safeway, Vons, Pavilions, Tom Thumb, and Randall's, among others. Upon consummation, the merger was expected to create the second-largest traditional grocery chain (next to Kroger) by store count and sales in the US, with approximately 2,400 stores \citep{ftc2015ceberus}. Figure \ref{figure:footprint} shows Albertsons and Safeway's footprint in the contiguous US in 2009.

\begin{figure}[hbt!]
\centering
\includegraphics[scale = 0.6]{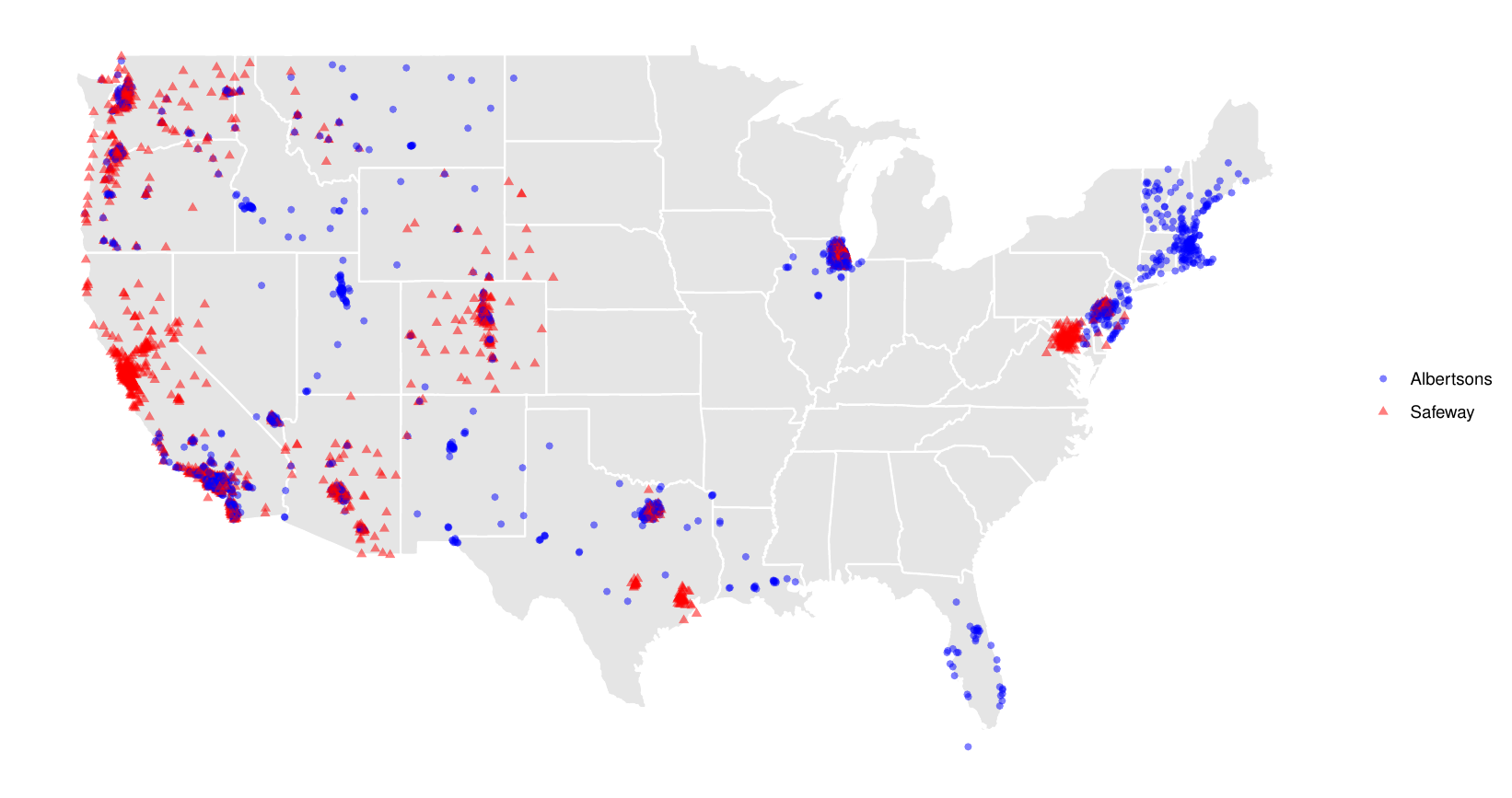}
\caption{Albertsons and Safeway footprint in 2009 \label{figure:footprint}}
\end{figure}

The FTC defined the relevant product market as supermarkets within ``hypermarkets.'' Supermarkets refer to ``traditional full-line retail grocery stores that sell, on a large-scale basis, food and non-food products that customers regularly consume at home---including, but not limited to, fresh meat, dairy products, frozen foods, beverages, bakery goods, dry groceries, detergents, and health and beauty products.'' Hypermarkets include chains such as Walmart Supercenters that sell an array of products not found in traditional supermarkets but also offer goods and services available at conventional supermarkets.

The FTC defined the relevant geographic markets as areas that range from a two- to ten-mile radius around each party's supermarkets, where the radius depends on factors such as population density, traffic, and unique market characteristics. The agency identified overlapping territories in Arizona, California, Colorado, Montana, Nevada, Oregon, Texas, Washington, and Wyoming. In late 2014, the FTC settled with the parties with a mandate to divest 168 stores in the overlap markets \citep{ftc2015ceberus}. Figure \ref{figure:divested} shows the locations of the divested stores.

\begin{figure}[hbt!]
\centering
\includegraphics[width = 10cm, height=5cm]{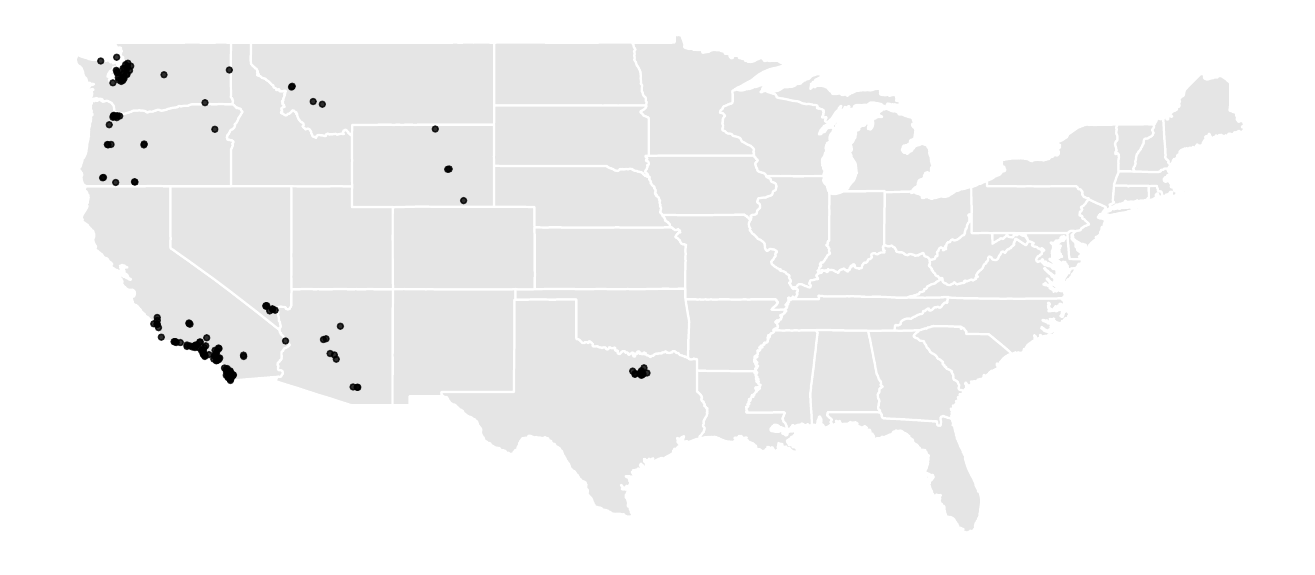}
\caption{Location of the 168 divested stores \label{figure:divested}}
\end{figure}

\subsection{Data}
I use the 2009 cross-section of AC Nielsen's (currently known as The Nielsen Company) Trade Dimensions TDLinx data for information on grocery stores' locations, sales, and characteristics in the US.\footnote{Year 2009 is the closest year before 2014 for which I can access the Trade Dimensions data. In 2013, SuperValue Inc. sold Albertsons, Acme, Jewel-Osco, Shaw's and Start Market banners to Cerberus' AB Acquision, the parent company of Albertsons Inc. To capture the Albertsons/Safeway merger investigation started in 2014, I assume Albertsons owns all the banners mentioned above, which my 2009 Trade Dimensions data encodes as SuperValu-owned. Note that AC Nielsen's Trade Dimensions data estimates weekly sales volume using a proprietary algorithm.} To generate conservative competitive effects estimates, I include a wide range of firms as potential competitors. Specifically, I include all grocery stores with selling space above 7,000 square feet but exclude military commissaries; my sample includes traditional supermarkets, supercenters, wholesale clubs, natural/gourmet stores, limited assortment stores, and warehouses. Since I cannot access confidential store-level margin data, I uniformly apply a relative margin of 0.27 to all stores based on the merging parties' 10-K reports around the time of the proposed merger. Finally, I obtain census tract-level demographic information from the IPUMS NHGIS database \citep{manson2023ipums}, including income, the proportion of the population with a college degree or higher, black population, and urbanicity (fraction of people living in census-designated urban areas). Table \ref{table:summary.statistics} reports the summary statistics.\footnote{I refer the readers to \citet{ellickson2020measuring} for more details on the grocery industry landscape.}

\begin{table}[htbp!]
\centering
\caption{Summary Statistics \label{table:summary.statistics}}
\scalebox{0.8}{
\begin{threeparttable}
\begin{tabular}{lccccc} 
\toprule
 & Mean & St. Dev. & 1st Quartile & Median & 3rd Quartile \\ 
\midrule
\emph{Tract Characteristics\tnote{1}} \\ [1ex]
Population (1,000) & 4.716 & 2.191 & 3.287 & 4.436 & 5.784\\
Median Household Income (\$1,000) & 60.926 & 29.741 & 39.925 & 54.314 & 75.108\\
College & 0.288 & 0.194 & 0.134 & 0.241 & 0.411\\
Black & 0.069 & 0.114 & 0.007 & 0.028 & 0.079\\
Urbanicity & 0.873 & 0.294 & 0.986 & 1.000 & 1.000\\ [1ex]
\emph{Store Characteristics} \\ [1ex]
Annual Revenue (\$1,000,000) & 18.123 & 18.301 & 6.518 & 11.732 & 22.161\\
Store Size (1,000 Sq. Ft.) & 30.359 & 17.875 & 16.000 & 28.000 & 38.000\\
Supermarket & 0.737 & 0.440 & 0.000 & 1.000 & 1.000\\
Supercenter & 0.101 & 0.301 & 0.000 & 0.000 & 0.000\\
Wholesale Club & 0.036 & 0.187 & 0.000 & 0.000 & 0.000\\
\addlinespace
Natural/Gourmet & 0.040 & 0.195 & 0.000 & 0.000 & 0.000\\
Limited Assortment & 0.072 & 0.259 & 0.000 & 0.000 & 0.000\\
Warehouse & 0.013 & 0.115 & 0.000 & 0.000 & 0.000\\
Big Chain & 0.789 & 0.408 & 1.000 & 1.000 & 1.000\\
Medium Chain & 0.131 & 0.338 & 0.000 & 0.000 & 0.000\\
\addlinespace
Small Chain & 0.080 & 0.271 & 0.000 & 0.000 & 0.000\\
\bottomrule
\end{tabular}
\begin{tablenotes}
\footnotesize
    \item[1] Sample consists of census tracts in overlap states.
    \item[2] I define big chains as those with over 100 stores. Medium chains are those with 10 to 100 stores. Small chains are those with 10 or fewer stores.
\end{tablenotes}
\end{threeparttable}
}
\end{table}

\subsection{Empirical Specification}
I assume that each firm is an owner of multiple stores and competes by setting a uni-dimensional price at each store. Each store corresponds to a ``product'' in the standard Bertrand-Nash framework. Price $p_j$ of store $j$ is interpreted as the unobserved price index. I estimate the gross upward pricing pressure indices of the merging parties' stores before and after the merger. I do not credit merger-specific efficiency (i.e., I set $\ddot{c}_j = 0$). Thus, the GUPPI at each store $j$ is given by $\mathit{GUPPI}_j = (1 + \epsilon_{jj}^{-1}) \sum_{k \in \mathcal{J}_B} m_k D_{j \to k}^R$, where $B$ represents the merger counterparty, and the own-price elasticity $\epsilon_{jj}$ is a known function of margins and revenue diversion ratios as shown in \eqref{equation:own.price.elasticity}. Given store-level margin data, the estimation of revenue diversion ratios completes the GUPPI calculation.

I estimate revenue diversion ratios using a parametric model based on \citet{ellickson2020measuring}, making the following assumptions. I define the relevant consumer unit as the census tract, effectively assuming a representative consumer per tract. Each representative consumer forms a consideration set $\mathcal{C}_i$ consisting of all stores within a 10-mile radius and has nested CES preferences. The set of nests $\mathcal{B}$ includes six categories: supermarket, supercenter, wholesale club, natural/gourmet, limited assortment, and the outside option. The budget share that consumer $i$ allocates to store $j$ in nest $b$ is
\[
\alpha_{ij} = s_b^i s_{j \vert b}^i, \; \text{ where } s_b^i = \frac{\exp(\mu_b I_{i,b})}{\sum_{q \in \mathcal{B}} \exp (\mu_q I_{i,q})} , \; s_{j \vert b}^i =  \frac{\exp(u_{ij} / \mu_b)}{\sum_{l \in \mathcal{C}_{i,b} } \exp(u_{il}/\mu_b) } ,
\]
$\mu_b \in [0,1]$ is the nesting parameter, $\mathcal{C}_{i,b}$ is the subset of stores in nest $b$, and $I_{i,b} \equiv \log \sum_{l \in \mathcal{C}_{i,b}} \exp(u_{il} / \mu_b)$ is the inclusive value of nest $b$.\footnote{If $\mu_b \to 1$ for all nests, the model collapses to the standard logit case. Consumers substitute only within each nest if $\mu_b \to 0$ for all nests.} The outside option forms its own nest with $\mu_0 = 1$. Following \citet{ellickson2020measuring}, who find that nesting parameters are similar across nests, I assume a common value $\mu_b = \mu$ for all $b \in B \backslash \{0\}$ to simplify estimation.

Latent utility $u_{ij}$ is specified as a linear projection onto observed covariates: $u_{ij} = x_{ij}^\top \theta$, where $x_{ij}$ includes the distance from tract to store (in miles), tract characteristics, and store characteristics. As in \citet{ellickson2020measuring}, I assume the control variables absorb unobserved price variation, though I adopt a more parsimonious set of controls for tractability. Finally, I assume that each tract allocates 13\% of its income to grocery spending based on \citet{ellickson2020measuring}'s estimates.\footnote{The grocery budget estimate is likely to be conservative. The 2014 Bureau of Labor Statistics Consumer Expenditure Survey reports that consumers spend approximately 12.6\% of their pre-tax income on food but only 7.4\% on food at home, which may be more relevant for estimating the grocery budget.} I estimate the utility parameters $\theta$ via nonlinear least squares by minimizing the distance between observed and model-implied store revenues.

\subsection{Estimation Results}

\paragraph{Utility Parameters}
Table \ref{table:utility.parameter.estimates} reports the utility parameter estimates from the nonlinear least squares problem. All coefficients have expected signs. For example, consumers dislike traveling far, especially more when they are in an urban area. Urban residents or richer/higher-educated people have lower grocery expenditures because they can substitute more for outside options such as restaurants or food delivery. Black consumers also spend less on groceries.\footnote{One explanation is that black neighborhoods have fewer supermarkets; see, e.g., \citet{bower2014intersection} and \citet{charron2017race}. The current framework does not capture how grocery store entries endogenously depend on neighborhood characteristics.} Consumers spend more in larger stores and at supercenters. Consumers value major banners such as Walmart, Costco, H-E-B, Whole Foods, and Trader Joe's. Finally, the nesting parameter $\mu$ is estimated to be 0.46, which indicates consumers perceive different grocery formats as highly differentiated.\footnote{My estimate of the nesting parameter is smaller than those found in \citet{ellickson2020measuring}, which reports $\mu \approx 0.75$ using 2006 TDLinx data. The differences may be attributed to multiple factors such as data years, model specification, and classification of nests.}

% \begin{table}[htbp!]
% \centering
% \caption{Utility parameter estimates \label{table:utility.parameter.estimates}}
% \scalebox{0.60}{
% \begin{tabular}{lS[table-format=3.2]S[table-format=2.2]} 
% \toprule
% & \text{Coef.} & \text{S.E.} \\ 
% \midrule
% Constant & 24.99 & (1.00)\\
% Distance & -0.03 & (0.01)\\
% Distance * Urbanicity & -0.05 & (0.01)\\
% Urbanicity & -0.24 & (0.10)\\
% log(Median HH Income) & -2.57 & (0.10)\\
% \addlinespace
% College & -1.71 & (0.16)\\
% Black & -4.87 & (0.27)\\
% log(Store Size) & 0.31 & (0.01)\\
% Supermarket & 0.03 & (0.02)\\
% Supercenter & 0.16 & (0.03)\\
% \addlinespace
% Wholesale Club & -0.22 & (0.04)\\
% Natural/Gourmet & -1.01 & (0.04)\\
% Limited Assortment & -1.76 & (0.05)\\
% Big Chain & -0.03 & (0.01)\\
% Medium Chain & 0.05 & (0.01)\\
% \addlinespace
% Albertsons & 0.23 & (0.01)\\
% Safeway & 0.32 & (0.01)\\
% Walmart & 0.46 & (0.02)\\
% Costco & 0.64 & (0.04)\\
% Kroger & 0.28 & (0.01)\\
% \addlinespace
% H-E-B & 0.52 & (0.03)\\
% Whole Foods & 0.53 & (0.03)\\
% Trader Joes & 0.63 & (0.03)\\
% Save Mart & 0.20 & (0.02)\\
% Winco & 0.44 & (0.03)\\
% \addlinespace
% Stater Bros & 0.47 & (0.02)\\
% Raleys & 0.29 & (0.02)\\
% Target & 0.08 & (0.03)\\
% $\mu$ & 0.46 & (0.01)\\
% \midrule
% Residual S.E. & \multicolumn{2}{c}{0.476} \\
% \bottomrule
% \end{tabular}
% }
% \end{table}

\begin{table}[htbp!]
\centering
\caption{Utility parameter estimates \label{table:utility.parameter.estimates}}
\scalebox{0.90}{
\begin{tabular}{lS[table-format=3.2]S[table-format=2.2] @{\hspace{3em}} lS[table-format=3.2]S[table-format=2.2]} 
\toprule
& \text{Coef.} & \text{S.E.} & & \text{Coef.} & \text{S.E.} \\ 
\midrule
Constant & 24.99 & (1.00) & Albertsons & 0.23 & (0.01) \\
Distance & -0.03 & (0.01) & Safeway & 0.32 & (0.01)\\
Distance * Urbanicity & -0.05 & (0.01) & Walmart & 0.46 & (0.02)\\
Urbanicity & -0.24 & (0.10) & Costco & 0.64 & (0.04)\\
log(Median HH Income) & -2.57 & (0.10) & Kroger & 0.28 & (0.01)\\
\addlinespace
College & -1.71 & (0.16) & H-E-B & 0.52 & (0.03)\\
Black & -4.87 & (0.27) & Whole Foods & 0.53 & (0.03)\\
log(Store Size) & 0.31 & (0.01) & Trader Joes & 0.63 & (0.03)\\
Supermarket & 0.03 & (0.02) & Save Mart & 0.20 & (0.02)\\
Supercenter & 0.16 & (0.03) & Winco & 0.44 & (0.03)\\
\addlinespace
Wholesale Club & -0.22 & (0.04) & Stater Bros & 0.47 & (0.02)\\
Natural/Gourmet & -1.01 & (0.04) & Raleys & 0.29 & (0.02)\\
Limited Assortment & -1.76 & (0.05) & Target & 0.08 & (0.03)\\
Big Chain & -0.03 & (0.01) & $\mu$ & 0.46 & (0.01)\\
Medium Chain & 0.05 & (0.01)\\
\midrule
Residual S.E. & 0.476 \\
\bottomrule
\end{tabular}
}
\end{table}

\paragraph{Revenue Diversion Ratios}
Figure \ref{figure:revenue.diversion.ratio} shows the distribution of revenue diversion ratios implied at the estimated utility parameter.\footnote{The figures' x-axes are truncated at 0.25 for presentation.} Figure \ref{figure:revenue.diversion.ratio.3.miles} reports the distribution of revenue diversion ratios for pairs of stores within 3 miles of each other. The revenue diversion ratios from one store to another tend to be small. However, Figure \ref{figure:aggregate.revenue.diversion.ratio} illustrates the total revenue diversion ratio from one store to \emph{all surrounding merger counterparty stores} can be significantly larger, indicating the importance of accounting for multi-store ownership in retail merger analysis. GUPPI statistics enable researchers to summarize how diversions to a network of surrounding stores aggregate, which is crucial for analyzing competition in the grocery industry.

\begin{figure}[hbt!]
\centering
\begin{subfigure}{0.9\textwidth}
\centering
\includegraphics[width = 0.8\textwidth]{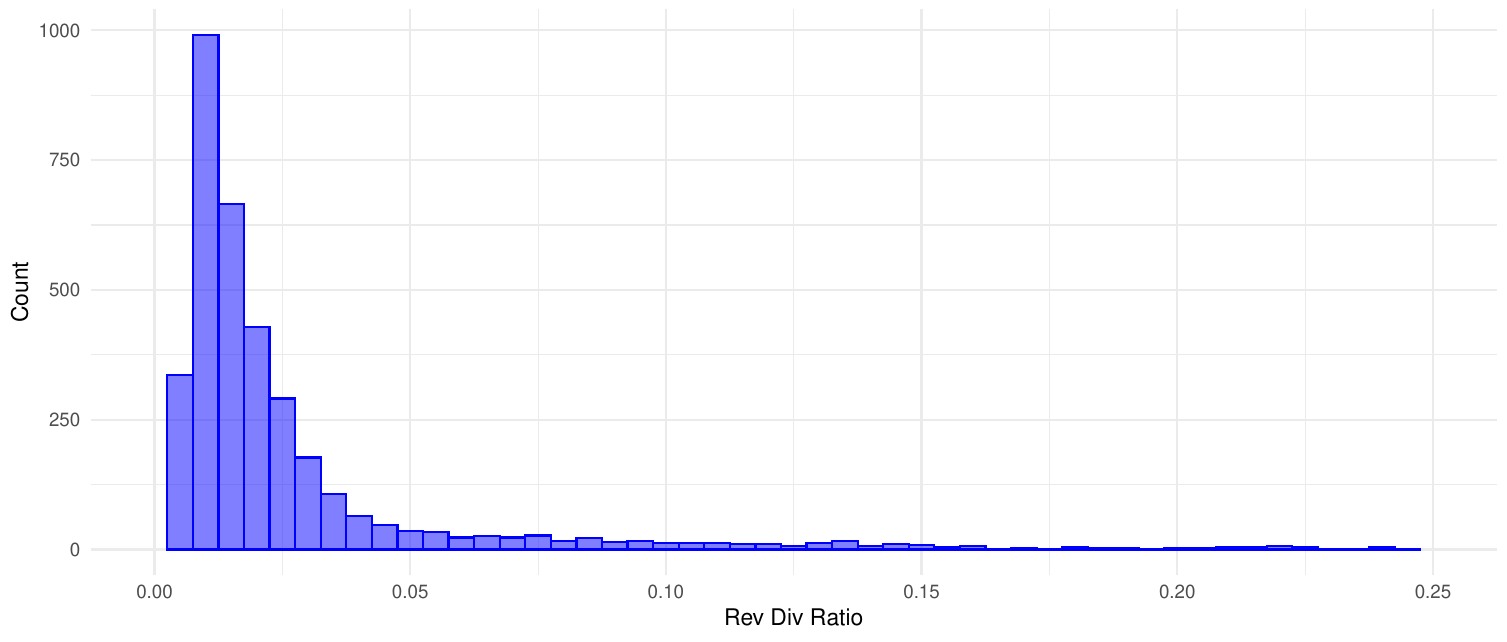}
\caption{Revenue diversion ratios to stores within 3 miles \label{figure:revenue.diversion.ratio.3.miles}}
\end{subfigure}
\begin{subfigure}{0.9\textwidth}
\centering
    \includegraphics[width = 0.8\textwidth]{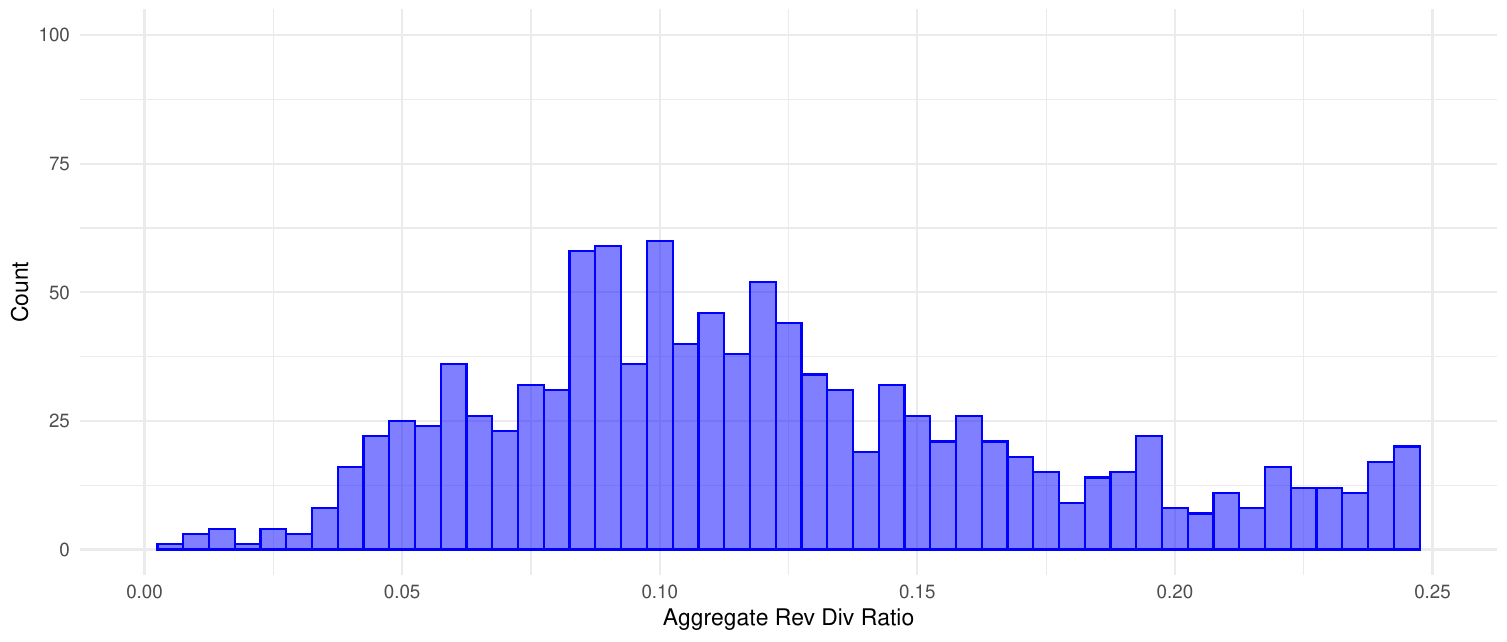}
    \caption{Aggregate revenue diversion ratios to merger counterparty stores \label{figure:aggregate.revenue.diversion.ratio}}
\end{subfigure}
\caption{Distribution of revenue diversion ratios\label{figure:revenue.diversion.ratio}}

\end{figure}

\paragraph{Gross Upward Pricing Pressure Indices}
The FTC-mandated divestiture contributed to a substantial reduction in the GUPPIs. Table \ref{table:guppi.store.count} reports the distribution of GUPPIs. Column ``Pre'' reports the store-level GUPPIs in the pre-divestiture regime. Column ``Post'' reports the GUPPIs of the remaining stores after the divestiture.\footnote{The table lists the divestiture of 165 stores, which is slightly less than the actual divestiture of 168 stores. This discrepancy of three stores arises from the time gap between my data, collected in 2009, and the year of the merger proposal in 2014.} Finally, Column ``Divested'' reports the GUPPIs of the divested stores prior to the merger. Overall, the divestiture significantly decreased the upward pricing pressures at many problematic stores; Figure \ref{figure:guppi.distribution} shows that the divestiture induced a leftward shift in the GUPPI distribution. The divested stores had relatively high GUPPIs in proportion compared to the overall distribution reported in the first column. Post divestiture, the number of high-GUPPI stores decreases substantially.

\begin{table}[htbp!]
\centering
\caption{Distribution of GUPPIs before and after the divestiture \label{table:guppi.store.count}}
\begin{tabular}{lccc} 
\toprule
GUPPI & Pre & Post & Divested \\ 
\midrule
0--1\% & 430 & 555 & 0\\
1--2\% & 348 & 440 & 9\\
2--3\% & 380 & 229 & 35\\
3--4\% & 195 & 113 & 32\\
4--5\% & 117 & 66 & 28\\
5\%-- & 168 & 63 & 61\\
\midrule
Total & 1,638 &  1,466 & 165 \\
\bottomrule
\end{tabular}
\end{table}

\begin{figure}[hbt!]
\centering
\includegraphics[width = 0.8\textwidth]{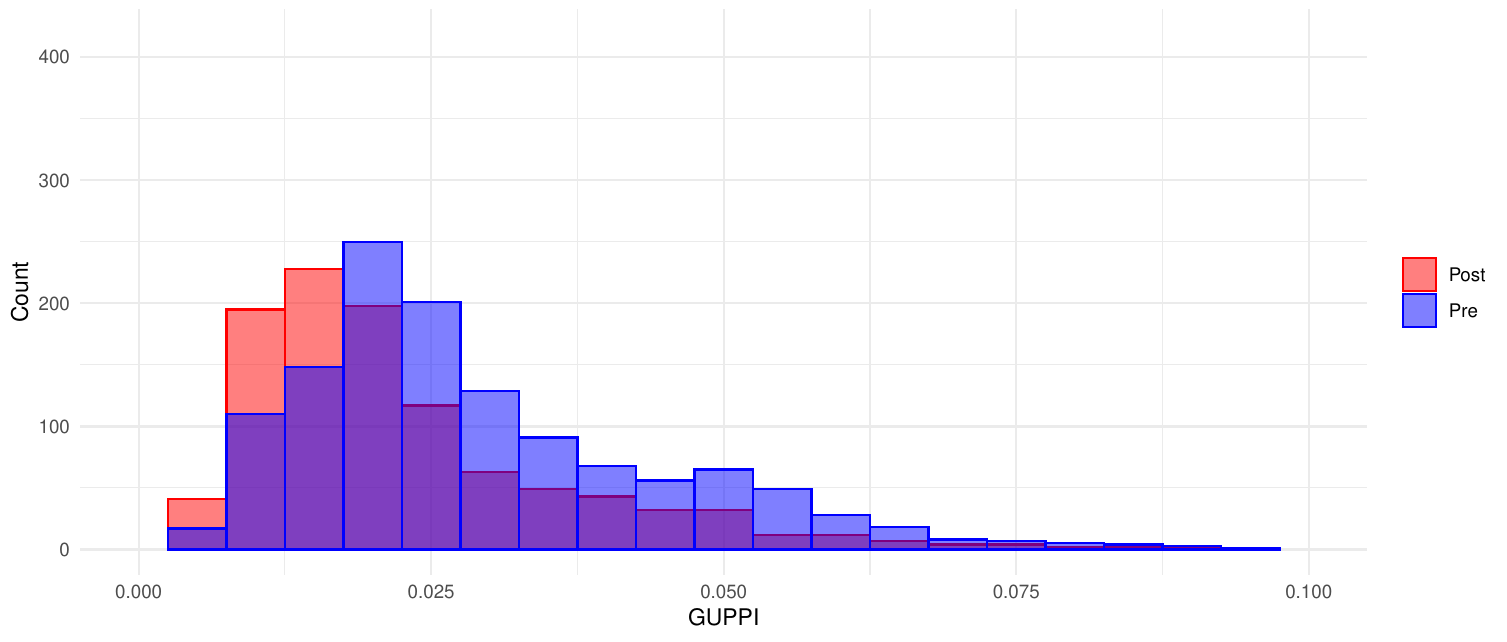}
\caption{Distribution of GUPPIs pre and post divestiture\label{figure:guppi.distribution}}
\end{figure}

\paragraph{Welfare Effects of Merger}
Based on \citet{miller2017upward}'s results, I approximate the merger pass-through matrix with an identity matrix, which amounts to predicting the merger price effects as $\ddot{p}_j \approx \mathit{GUPPI}_j$.\footnote{Although it may be possible to calculate the merger pass-through matrix using the nested CES demand assumption, I use \citet{miller2017upward}'s approach to avoid calculating a high-dimensional matrix and generate conservative estimates of merger price effects.} To calculate the total merger harm to consumers, I use equation \eqref{equation:total.welfare.effect}, where each $\Delta \mathit{CS}_j \approx - \ddot{p}_j R_j$ (Lemma \ref{lemma:approximation.of.welfare.effects}) represents the predicted consumer harm at merging parties' store $j$. The estimated annual consumer harm before the divestiture is \$621 million. Post-merger, the annual consumer harm reduces to \$383 million. 

\paragraph{Compensating Cost Efficiencies to Offset Upward Pricing Pressures}
The consumer harm remains substantial because I assume no efficiency credit (i.e., $\ddot{c}_j = 0$). If the merger induces reductions in marginal costs, then the consumer harm will be smaller. In practice, it is common for antitrust authorities to evaluate merger effects after crediting some degree of marginal cost efficiencies. Like \citet{nocke2022concentration}, I study the antitrust authorities' beliefs in firms' merger-induced marginal cost reductions by estimating the percentage change in marginal costs $\ddot{c}_j$'s that offset the upward pricing pressures.\footnote{
From \eqref{equation:gross.upward.pricing.pressure}, the net gross upward pricing pressure index is $\mathit{GUPPI}_j = \mathit{GUPPI}_j^\text{no credit} + \ddot{c}_j ( 1- m_j)$. Thus, the percentage reduction in marginal cost to offset upward pricing pressure at store $j$ is $-\ddot{c}_j = \mathit{GUPPI}_j^\text{no credit} / (1-m_j)$. This measure is easier to calculate than \citet{werden1996robust}'s compensating marginal cost reductions, which account for all firms' optimal pricing equations simultaneously.}

Figure \ref{figure:cost.efficiency.credit.distribution} reports the distribution of cost-efficiency credits required to offset upward pricing pressures at the merging firms' stores before and after the divestiture. The 90th, 95th, and 99th quantiles of the post-divestiture compensating cost efficiencies are 5.2\%, 6.6\%, and 9.9\%, respectively. If the FTC hypothetically were to credit cost efficiencies of 5\%, 6\%, or 7\% to all stores, it would predict no price increase at more than 90\%, 95\%, or 99\% of the remaining stores, respectively.

\begin{figure}[hbt!]
\centering
\includegraphics[width = 0.8\textwidth]{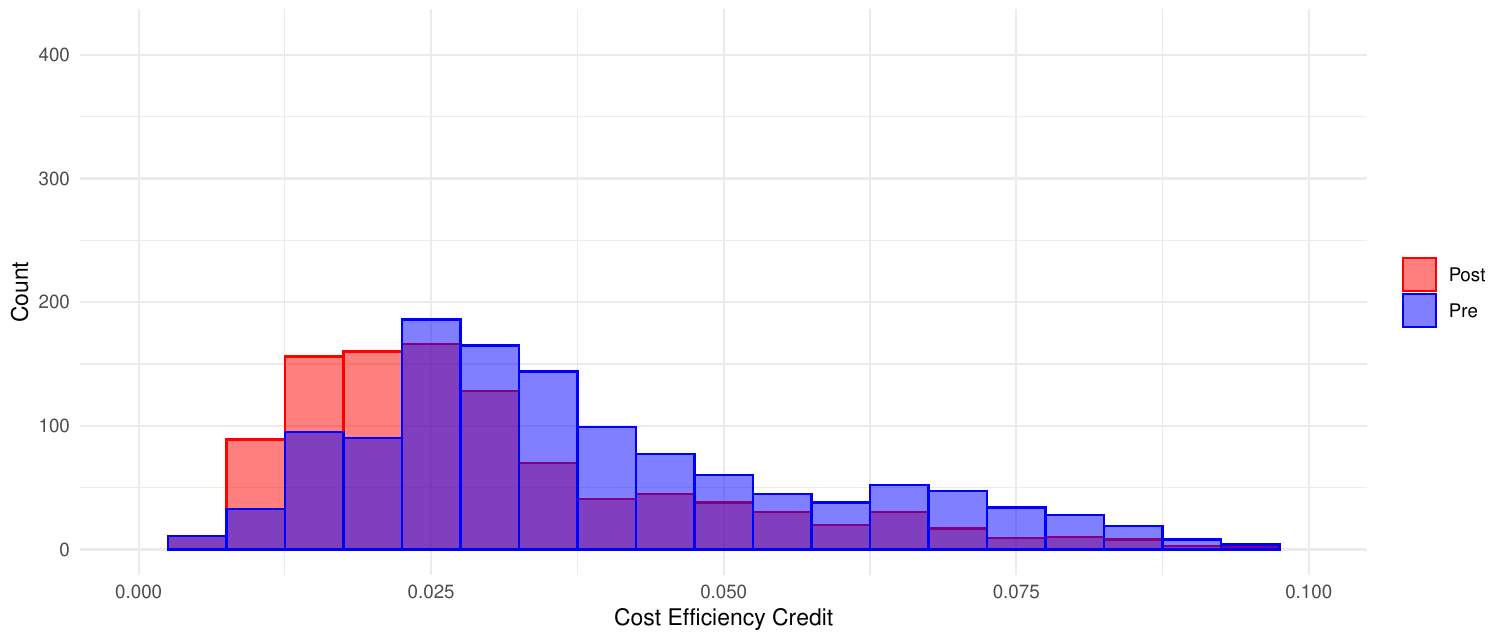}
\caption{Distribution of compensating cost efficiencies pre and post divestiture\label{figure:cost.efficiency.credit.distribution}}
\end{figure}
\subsection{Limitations}

My analysis faces several limitations. First, the data used in this application are more limited than what is typically available in actual merger reviews. For instance, TDLinx relies on a proprietary and undisclosed imputation method to estimate revenues. By contrast, antitrust authorities can obtain actual store-level sales and margin data or detailed consumer-level transaction data from firms operating in the relevant markets \citep{hosken2016horizontal}. Access to such data would improve the accuracy of both revenue diversion ratio estimates and the consumer budget share parameter.

Second, the GUPPI estimates can be sensitive to modeling assumptions, particularly the specification of the nesting structure, which affects the estimated diversion ratios. In practice, disputes between antitrust authorities and merging parties often hinge on market definition, which in turn shapes assumptions about consumers’ budgets and substitution patterns. For example, in \emph{Whole Foods/Wild Oats}, the FTC defined the market as ``premium, natural, and organic supermarkets,'' which is narrower than the broader market definition used in this analysis. While a flexible random coefficients model could produce more robust estimates under weaker assumptions, such an approach typically requires access to price data, which is not available here.

Finally, I do not explicitly address the tension between the stylized fact that grocery chains often use uniform pricing within zones and my modeling assumption of store-level pricing. Nevertheless, the store-level pricing assumption provides a reasonable baseline for several reasons. First, although prices are generally more uniform within chains than across them, prior research documents substantial price variation across local markets \citep{dellavigna2019uniform, hitsch2021prices, aparicio2024pricing}.\footnote{If there is large within-chain price variation across local markets, the researcher may want to add extra variables that can control for unobserved prices (e.g., regional fixed effects) in the demand estimation stage.} Second, uniform pricing may reflect a strategic simplification rather than an inability to price locally, especially when the gains from granular pricing are modest. Third, chains appear capable of highly localized pricing; for example, Kroger reportedly employs ``micro-zones'' encompassing as few as one to five stores for staple items such as milk, bananas, and eggs \citep{ftc2024kroger, FTC_v_Kroger_PFF_2024}. While it is theoretically possible to model uniform pricing within a Bertrand-Nash framework, applying the idea for empirical analysis is challenging in practice because pricing zone definitions are typically unobserved.

\section{Conclusion \label{section:99.conclusion}}
Throughout this paper, I have illustrated that a lack of price data is not a barrier to conducting a thorough unilateral effects analysis for horizontal mergers. Data on revenues, profit margins, and revenue diversion ratios are sufficient to identify the price and welfare implications of mergers. When supplemented with additional assumptions about consumer demand, revenue diversion ratios become identifiable from cross-sectional data on consumers' expenditures, and merger simulations are feasible. The approach detailed in this study offers broad applicability to various industries in scenarios where access to price data is limited.

% \todo{Aviv has a paper that criticizes the CES preference as imposing a priori assumption on substitution pattern. This is a challenge.}

I see several avenues for future research. First, finding alternative methods for estimating the key statistics used in my framework---diversion ratios, margins, and pass-through rates---will be interesting as these statistics are not always easy to obtain. Econometric approaches that can relax the homogeneous price responsiveness assumption to estimate diversion ratios will be helpful. Furthermore, developing alternative methods to obtain reliable pass-through rates will increase the attractiveness of the first-order approach to merger analysis.

Second, developing empirical and simulation results under various data scenarios would be valuable. In many applications, researchers have access only to partial information, such as prices, margins, or diversion ratios aggregated across multiple products. While I have abstracted from these complications to focus on identification, important questions remain about how to make use of partial price and margin data and how aggregation error may affect empirical findings.

Finally, it is important to develop empirical merger analysis frameworks tailored to settings beyond standard Bertrand-Nash competition, such as bargaining, auctions, nonlinear or dynamic pricing, and collusion. These alternative competitive environments often arise in real-world markets and may require fundamentally different modeling approaches to capture the strategic interactions and merger effects.

%%% REFERENCE %%%
% \clearpage
\singlespacing
\bibliographystyle{apalike}
\bibliography{references}
\doublespacing

\clearpage
\appendix
\part*{Appendix}
\section{Proofs \label{section:proof}}
\subsection{Proof of Lemma \ref{lemma:relationship}}

\paragraph{Proof of Lemma \ref{lemma:relationship}.\ref{lemma:relationship.item.1}}
Recall $\epsilon_{jj} = \frac{\partial q_j}{\partial p_j} \frac{p_j}{q_j}$, $\epsilon_{kj} = \frac{\partial q_k}{\partial p_j} \frac{p_j}{q_k}$, $\epsilon_{jj}^R = \frac{\partial R_j}{\partial p_j} \frac{p_j}{R_j}$, $\epsilon_{kj}^R = \frac{\partial R_k}{\partial p_j} \frac{p_j}{R_k}$. Using these objects, it is straightforward to verify $D_{j \to k} = - \frac{\epsilon_{kj}}{\epsilon_{jj}} \frac{q_k}{q_j}$ and $D_{j \to k}^R = - \frac{\epsilon_{kj}^R}{\epsilon_{jj}^R} \frac{R_k}{R_j}$. \qed

\paragraph{Proof of Lemma \ref{lemma:relationship}.\ref{lemma:relationship.item.2}}
From $R_k = p_k q_k$, $\frac{\partial R_k}{\partial p_j} = q_j \mathbb{I}_{j=k} + p_k \frac{\partial q_k}{\partial p_j}$ for arbitrary $j$ and $k$. Then for arbitrary $j$ and $k$, $\epsilon_{kj}^R = \frac{\partial R_k}{\partial p_j} \frac{p_j}{R_k}
  = \mathbb{I}_{j = k} q_j \frac{p_j}{R_k} + p_k \frac{\partial q_k}{\partial p_j} \frac{p_j}{R_k} = \mathbb{I}_{j = k} + \frac{\partial q_k}{\partial p_j} \frac{p_j}{q_k} = \mathbb{I}_{j = k} + \epsilon_{kj}$. \qed

\paragraph{Proof of Lemma \ref{lemma:relationship}.\ref{lemma:relationship.item.3}}

    Finally, by definition, $R_k = p_k q_k$, so 
$\frac{\partial R_k}{\partial p_j} = \mathbb{I}_{j = k} q_k + p_k \frac{\partial q_k}{\partial p_j}$ for arbitrary $j$ and $k$. Then
\begin{align*}
D_{j \to k}^R & = - \frac{\mathbb{I}_{j = k}q_k + p_k \frac{\partial q_k}{\partial p_j}}{q_j + p_j \frac{\partial q_j}{\partial p_j}} \\
& = - \frac{ \mathbb{I}_{j = k} (\frac{q_k}{p_j})/ (\frac{\partial q_k}{\partial p_j})  +  \frac{p_k}{p_j}(\frac{\partial q_k}{\partial p_j})/( \frac{\partial q_j}{\partial p_j}) }{(\frac{q_j}{p_j})/( \frac{\partial q_j}{\partial p_j}) + 1} \\
& = -\frac{\mathbb{I}_{j = k} \epsilon_{kj}^{-1} - (\frac{p_k}{p_j}) D_{j \to k}}{\epsilon_{jj}^{-1} + 1}.
\end{align*}
Rewriting the last line gives $(1 + \epsilon_{jj}^{-1}) D_{j \to k}^R + \mathbb{I}_{j = k} \epsilon_{kj}^{-1} = D_{j \to k} \frac{p_k}{p_j}$. But as $\mathbb{I}_{j = k} \epsilon_{kj}^{-1} = \mathbb{I}_{j = k} \epsilon_{jj}^{-1}$, we have $(1 + \epsilon_{jj}^{-1})D_{j \to k}^R + \mathbb{I}_{j = k} \epsilon_{jj}^{-1} = D_{j \to k} \frac{p_k}{p_j}$.
\qed

\subsection{Proof of Lemma \ref{lemma:own.price.elasticity}}
Solving equation \eqref{equation:first.order.conditions.rewritten} for $\epsilon_{jj}$ gives the desired expression \eqref{equation:own.price.elasticity}. \qed

\subsection{Proof of Proposition \ref{proposition:guppis.are.identified}}

The statement follows because cost efficiencies $\ddot{c}_j$, margins $m_j$, revenue diversion ratios $D_{j\to k}^R$ are observed, and the own-price elasticity $\epsilon_{jj}$ is a known function of margins and revenue diversion ratios.
\qed
\subsection{Proof of Proposition \ref{proposition:upp.and.merger.price.effects}}

    The statement follows because the first three assumptions (Assumptions \ref{assumption:data}, \ref{assumption:revenue.diversion.ratios.data}, and \ref{assumption:cost.savings}) ensure the identification of GUPPIs, and the last assumption (Assumption \ref{assumption:upp.approximates.true.price.effect}) ensures that GUPPIs can be translated to first-order approximation of merger price effects. \qed

\subsection{Proof of Lemma \ref{lemma:approximation.of.welfare.effects}}
First, the consumer surplus variation is approximated as $\Delta \mathit{CS}_j \approx - \Delta p_j q_j^0 = -(\Delta p_j/p_j^0) \times (p_j^0 q_j^0) = - \ddot{p}_j R_j^0 $. Next, consider the change in producer surplus $\Delta \mathit{PS}_j = \Delta p_j \times q_j^1 + \Delta q_j \times (p_j^0 - c_j^0) - \Delta c_j\times q_j^1$. The post-merger quantity demanded can be approximated as $q_j^1 \approx q_j^0 + (\partial q_j / \partial p_j) (p_j^1 - p_j^0) = q_j^0 (1 + \epsilon_{jj} \ddot{p}_j)$, which gives $\Delta q_j \equiv q_j^1 - q_j^0 \approx \epsilon_{jj} q_j^0 \ddot{p}_j$. Plugging them into $\Delta \mathit{PS}$ above gives $\Delta \mathit{PS}_j \approx \ddot{p}_j R_j^0 (1+\epsilon_{jj} \ddot{p}_j) + \epsilon_{jj} R_j^0 \ddot{p}_j m_j^0 - \frac{\Delta c_j}{p_j^0} R_j^0 (1+\epsilon_{jj} \ddot{p}_j)$. Plugging in 
$\Delta c_j/p_j = \frac{\Delta c_j}{c_j^0} \frac{c_j^0}{p_j^0} = \ddot{c}_j (1-m_j^0)$ and rearranging gives the desired expression $\Delta \mathit{PS}_j \approx (\ddot{p}_j - \ddot{c}_j (1-m_j)) R_j ( 1 + \epsilon_{jj} \ddot{p}_j) + \epsilon_{jj} R_j \ddot{p}_j m_j$. \qed
\subsection{Proof of Proposition \ref{proposition:consumer.welfare}}

The assumptions ensure that the own-price elasticities of demand and merger price effects are identified. Combining them with data on revenues, margins, and cost efficiencies, $\Delta \mathit{CS}_j$'s and $\Delta \mathit{PS}_j$'s are identified. \qed

\subsection{Proof of Lemma \ref{lemma:cmcr}}

First, using the fact that $m_j = 1 - \frac{c_j}{p_j}$, the compensating marginal cost reductions can be expressed in terms of margins as $\ddot{c}_j = \frac{c_j^1 - c_j^0}{c_j^0} = \frac{(c_j^1 / p_j^1) (p_j^1 / p_j^0) - c_j^0 / p_j^0}{c_j^0 / p_j^0} = \frac{(1 - m_j^1) - (1 - m_j^0)}{1 - m_j^0} = \frac{m_j^0 - m_j^1}{1 - m_j^0}$; note that I have used $p_j^1 = p_j^0$ as pre- and post-merger prices should be equal by the definition of CMCR. Next, expressing the post-merger FOCs in terms of own-price elasticities of demand, margins, and revenue diversion ratios, but evaluating the own-price elasticities of and revenue diversion ratios at their pre-merger values, gives \eqref{equation:post.merger.foc.for.cmcr}, which implicitly defines the post-merger margins. \qed
\subsection{Proof of Proposition \ref{proposition:cmcr}}

If Assumptions \ref{assumption:data} and \ref{assumption:revenue.diversion.ratios.data} hold, then the only unknowns in equation \eqref{equation:post.merger.foc.for.cmcr}. Solving \eqref{equation:post.merger.foc.for.cmcr} for post-merger margins $m_j^1$'s and plugging them into \eqref{equation:cmcr.calculation.from.margin} gives the CMCRs, as desired. \qed

\subsection{Proof of Lemma \ref{lemma:revenue.diversion.ratio.under.ces.utility}}

I omit the derivation of the multinomial logit functional form for $\alpha_{ij}$ as it is standard.

Next, to obtain \eqref{equation:ces.own.price.elasticity.of.revenue}, consider the revenue equation $R_j = \int_{i \in \mathcal{S}_j} \alpha_{ij} B_i d\mu$. Differentiating both sides with respect to $p_j$ gives $\frac{\partial R_j}{\partial p_j} = \int_{i \in \mathcal{S}_j} \frac{\partial \alpha_{ij}}{\partial u_{ij}} \frac{\partial u_{ij}}{\partial p_j} B_i d\mu = \int_{i \in \mathcal{S}_j} \alpha_{ij}(1 - \alpha_{ij}) \frac{1-\eta}{p_j}B_i d\mu $. Thus, $\epsilon_{jj}^R = \frac{\partial R_j}{\partial p_j} \frac{p_j}{R_j} = \frac{1-\eta}{R_j} \int_{i \in \mathcal{S}_j} \alpha_{ij} ( 1-\alpha_{ij}) B_i d\mu $. Noting that $R_j = \int_{\tilde{i} \in \mathcal{S}_j} \alpha_{\tilde{i}j} B_{\tilde{i}} d\mu$ and $\epsilon_{i,jj}^e = \frac{\partial e_{ij}}{\partial p_j} \frac{p_j}{e_{ij}} = (1-\eta)(1-\alpha_{ij})$, the own-price elasticity of revenue can be rewritten as $\epsilon_{jj}^R = \int_{i \in \mathcal{S}_j} \overline{w}_{ij} \epsilon_{i,jj}^e d\mu$, where $\overline{w}_{ij} = \frac{\alpha_{ij} B_i}{\int_{\tilde{i} \in \mathcal{S}_j} \alpha_{\tilde{i}j} B_{\tilde{i}} d\mu}$.

Finally, to obtain \eqref{equation:ces.revenue.diversion.ratio}, note that the revenue diversion ratio from product $j$ to $k$, assuming that Leibniz's rule applies, is
\begin{align*}
D_{jk}^R & = - \frac{
\int_{i \in \mathcal{S}_k} \partial e_{ik} / \partial p_j d\mu
}{
\int_{\tilde{i} \in \mathcal{S}_j} \partial e_{\tilde{i}j} / \partial p_j d\mu
} \\ 
& = - \int_{i \in \mathcal{S}_k} \left( \frac{\partial e_{ik} / \partial p_j}{\int_{i \in \mathcal{S}_j} \partial e_{ij} / \partial p_j d\mu} \right) d\mu \\ 
& = \int_{i \in \mathcal{S}_k} \left( \frac{\partial e_{ij} / \partial p_j}{\int_{i \in \mathcal{S}_j} \partial e_{ij} / \partial p_j d\mu} \right) \left(- \frac{\partial e_{ik} / \partial p_j}{\partial e_{ij} / \partial p_j} \right) d\mu.
\end{align*}
Plugging in $\frac{\partial e_{ij}}{\partial p_j} = \frac{\partial \alpha_{ij}}{\partial u_{ij}} \frac{\partial u_{ij}}{\partial p_j} B_i = \alpha_{ij}(1-\alpha_{ij})\frac{1-\eta}{p_j}B_i$ and $\frac{\partial e_{ik}}{\partial p_j} = \frac{\partial \alpha_{ik}}{\partial u_{ij}} \frac{\partial u_{ij}}{\partial p_j} B_i = -\alpha_{ij}\alpha_{ik} \frac{1-\eta}{p_j} B_i$ and simplifying (canceling out $\frac{1-\eta}{p_j}$ from the numerators and denominators) gives the desired expression. \qed

\subsection{Proof of Proposition \ref{proposition:identification.of.diversion.ratios.under.ces.utility}}

Consumers' budgets are assumed to be known. Then, suppose consumers' expenditure shares for the merging firms' products are known. In that case, these shares can be plugged into \eqref{equation:ces.revenue.diversion.ratio} to calculate the revenue diversion ratios. \qed

\subsection{Proof of Proposition \ref{proposition:identification.with.second.choice.data}}

Let $v_{ij} \equiv \exp(u_{ij})$. The revenue of product $k$ before the removal of product $j$ is
$R_k^\text{pre} = \int_{i \in \mathcal{S}_k} \alpha_{ik}^\text{pre} B_i d\mu$
where $\alpha_{ij}^\text{pre} = \frac{v_{ik}}{\sum_{l \in \mathcal{C}_i} v_{il}}$.
After removing product $j$, the revenue of product $k$ becomes $R_k^\text{post} = \int_{i \in \mathcal{S}_k} \alpha_{ik}^\text{post} B_i d\mu$ with $\alpha_{ik}^\text{post} = \frac{v_{ik}}{\sum_{l \in \mathcal{C}_i \backslash j} v_{il}}$. The revenue diversion ratio associated with the removal of product $j$ is $D_{jk}^{\text{2nd}} = - \frac{R_k^\text{post} - R_k^\text{pre}}{0 - R_j^\text{pre}}  = \frac{\int_{i \in \mathcal{S}_k} (\alpha_{ik}^\text{post} - \alpha_{ik}^\text{pre}) B_i d\mu}{\int_{i \in \mathcal{S}_j} \alpha_{ij}^\text{pre} B_i d\mu} = \int_{i \in \mathcal{S}_k} \left(\frac{\alpha_{ij}^\text{pre} B_i}{\int_{i \in \mathcal{S}_j} \alpha_{ij}^\text{pre} B_i d\mu}\right) \left( \frac{\alpha_{ik}^\text{post} - \alpha_{ik}^\text{pre}}{\alpha_{ij}^\text{pre}} \right) d\mu$. Thus, it is sufficient to verify that $\frac{\alpha_{ik}^\text{post} - \alpha_{ik}^\text{pre}}{\alpha_{ij}^\text{pre}} = \frac{\alpha_{ik}^\text{pre}}{1 - \alpha_{ij}^\text{pre}}$.
Observe
\begin{align*}
    \frac{\alpha_{ik}^\text{post} - \alpha_{ik}^\text{pre}}{\alpha_{ij}^\text{pre}} & = \frac{\frac{v_{ik}}{\sum_{l \in \mathcal{C}_i \backslash j} v_{il}} - \frac{v_{ik}}{\sum_{l \in \mathcal{C}_i } v_{il}}}{ \frac{v_{ij}}{\sum_{l \in \mathcal{C}_i} v_{il}}} = \frac{ \left( \frac{\sum_{l \in \mathcal{C}_i } v_{il}}{\sum_{l \in \mathcal{C}_i \backslash j} v_{il}} \right)v_{ik} - v_{ik}}{v_{ij}} \\
    & = \frac{v_{ik}}{v_{ij}} \left( \frac{v_{ij}}{\sum_{l \in \mathcal{C}_i \backslash j} v_{il} } \right) = \frac{ \frac{v_{ik}}{\sum_{l \in \mathcal{C}_i} v_{il}}}{ 1 - \frac{ v_{ij}}{\sum_{l \in \mathcal{C}_i} v_{il}}},
\end{align*}
which is what I wanted to show. \qed

\subsection{Proof of Lemma \ref{lemma:merger.simulation.under.ces}}

Recall that the latent utility function is specified as $u_{ij} = \log \beta_{ij} + (1- \eta) \log p_j$. Since $\beta_{ij}$ is fixed (unaffected by price change) and $p_j^\text{post} \equiv p_j^\text{pre} ( 1 + \ddot{p}_j)$, the post-merger utility indices can be written as $u_{ij}^\text{post} = \log \beta_{ij}^\text{pre} + (1-\eta) \log p_j^\text{post} = \log \beta_{ij}^\text{pre} + (1-\eta) \log p_j^\text{pre} + (1-\eta) \log (1+ \ddot{p}_j) = u_{ij}^\text{pre} + (1-\eta) \log (1 + \ddot{p}_j)$. The post-merger choice probabilities $\alpha_{ij}^\text{post}$ are soft-max functions of the post-merger utility indices $(u_{ij}^\text{post})_{j \in \mathcal{C}_i}$. Next, the post-merger own-price elasticities of demand and the revenue diversion ratios use the characterizations \eqref{equation:ces.own.price.elasticity.of.revenue} (with the fact that $\epsilon_{jj} = \epsilon_{jj}^R - 1)$ and \eqref{equation:ces.revenue.diversion.ratio}. Finally, $m_j^\text{post} = \frac{p_j^\text{post} - p_j^\text{pre}}{p_j^\text{pre}} = 1 - \frac{c_j^\text{post}}{p_j^\text{post}} = 1 - \frac{c_j^\text{post} - c_j^\text{pre} + c_j^\text{pre}}{p_j^\text{post}} = 1 - \frac{c_j^\text{post} - c_j^\text{pre}}{c_j^\text{pre}} \frac{c_j^\text{pre}}{p_j^\text{pre}} \frac{p_j^\text{pre}}{p_j^\text{post}} - \frac{c_j^\text{pre}}{p_j^\text{pre}} \frac{p_j^\text{pre}}{p_j^\text{post}} = 1 - \frac{c_j^\text{pre}}{p_j^\text{pre}} \frac{p_j^\text{pre}}{p_j^\text{post}} (1 + \ddot{c}_j) = 1 - (1-m_j^\text{pre}) \frac{1 + \ddot{c}_j}{1 + \ddot{p}_j}$. \qed
\subsection{Proof of Proposition \ref{proposition:merger.simulation.under.ces}}

To make the components of the post-merger first-order conditions (consumer expenditure shares, own-price elasticities of demand, and revenue diversion ratios) known functions of arbitrary relative price change $\ddot{p}$ under Lemma \ref{lemma:merger.simulation.under.ces}, the analyst must know the pre-merger utility indices $u_{ij}^\text{pre}$ and the elasticity of substitution parameter $\eta$. $u_{ij}^\text{pre}$'s can be identified by inverting the soft-max expenditure share formula $\alpha_{ij}^\text{pre} = \frac{\exp(u_{ij}^\text{pre})}{\sum_{l \in \mathcal{C}_i} \exp(u_{il}^\text{pre})}$. Next, the parameter $\eta$ can be identified by leveraging \eqref{equation:ces.own.price.elasticity.of.revenue}, $\epsilon_{jj} = \epsilon_{jj}^R - 1$, and the fact that the own-price elasticities of demand can be identified. \qed

\section{Identification with Discrete-Continuous Demand \label{section:identification.with.discrete.continuous.demand}}

The main body of the paper demonstrates how the CES preference assumption simplifies the estimation of revenue diversion ratios and facilitates merger simulation. In this section, I extend the analysis by establishing identification conditions for a broader class of discrete-continuous demand models, of which CES is a special case. I rely on the following assumption. 

\begin{assumptionM}[Discrete-continuous demand] \label{assumption:discrete.continuous.demand} Let $u_{ij}^* = u_{ij} + \varepsilon_{ij}$ be the latent utility of consumer $i$ from consuming product $j$ per unit of expenditure, where $u_{ij}$ is the deterministic component, and $\varepsilon_{ij}$ is the stochastic component of latent utility with absolutely continuous full support distribution $\mathbb{R}^{\dim (\mathcal{C}_i)}$. Assume that the latent utility from the outside option is normalized to $u_{i0} = 0$. 
\begin{enumerate}
    \item \label{assumption:discrete.continuous.demand.1} For each consumer $i$ and product $j$, $\alpha_{ij} = \mathbb{E}[\mathbb{I} \{ j = \arg \max_{l \in \mathcal{C}_i} u_{il}^*\}]$. 
    \item \label{assumption:discrete.continuous.demand.2} For each consumer $i$ and product $j$, $u_{ij} = \delta_{ij} + (1-\eta) \log p_j$, where $\delta_{ij}$ is the latent utility from non-price characteristics of product $j$, and $\eta > 1$ is the price responsiveness parameter. 
\end{enumerate}
\end{assumptionM}

Recall that consumer $i$'s expenditure on product $j$ is given by $e_{ij} = \alpha_{ij} B_i$. Assumption \ref{assumption:discrete.continuous.demand}.\ref{assumption:discrete.continuous.demand.1} allows the analyst to model consumer behavior as if each unit of the budget is allocated through a discrete choice governed by an additive random utility model. CES utility is a special case in which expenditure shares take the familiar softmax form $\alpha_{ij} = \frac{u_{ij}}{\sum_{l \in \mathcal{C}_i} u_{il}}$. Assumption \ref{assumption:discrete.continuous.demand}.\ref{assumption:discrete.continuous.demand.2} serves three purposes: (i) it imposes an exclusion restriction, allowing price $p_j$ to affect expenditure only through $u_{ij}$; (ii) it assumes homogeneous price responsiveness to simplify identification; and (iii) it imposes log-linearity in price to facilitate merger simulation.\footnote{Log-linearity is not strictly necessary for establishing the identification arguments for revenue diversion ratios. However, such specification is standard. Furthermore, \citet{dube2022discrete} shows that log-linearity is crucial for establishing the Hurwicz-Uzawa integrability of the demand system.} In sum, Assumption \ref{assumption:discrete.continuous.demand} encompasses a broad class of discrete-continuous additive random utility models, such as nested CES, but excludes models with heterogeneous price sensitivity, such as random coefficient models.\footnote{I refer the reader to \citet{dube2022discrete} for a microeconomic foundation of Assumption \ref{assumption:discrete.continuous.demand}.}

\subsection{Identification of Revenue Diversion Ratios \label{subsection:revenue.diversion.ratios.under.discrete.continuous.demand}} 

For an arbitrary demand function, the identity $R_j = \int_{i \in \mathcal{S}_j} e_{ij} d\mu$ implies that the revenue diversion ratio can be expressed as a weighted average of individual-level expenditure diversion ratios.

\begin{lemma}[Revenue diversion ratio as a weighted sum of individual expenditure diversion ratio] \label{lemma:revenue.diversion.ratio.as.a.weighted.sum.general.case} Assuming that Leibniz's rule applies, the revenue diversion ratio for $j \neq k$ can be written as
\[
D_{j \to k}^R = \int_{i \in \mathcal{S}_k} w_{ij} D_{i,j \to k}^e d\mu, \; \text{ where } w_{ij} = \frac{\partial e_{ij} / \partial p_j}{\int_{\tilde{i} \in \mathcal{S}_j} \partial e_{\tilde{i} j} / \partial p_j d\mu}, \; D_{i, j \to k}^e = - \frac{\partial e_{ik} / \partial p_j}{\partial e_{ij} / \partial p_j}.  \]
\end{lemma}
\begin{proof}
    The revenue diversion ratio from product $j$ to $k$, assuming that Leibniz's rule applies, is $D_{j \to k}^R = - \frac{\partial R_k / \partial p_j}{\partial R_j / \partial p_j} = - \frac{\int_{i \in \mathcal{S}_k} \partial e_{ik} / \partial p_j d\mu }{\int_{\tilde{i} \in \mathcal{S}_j} \partial e_{\tilde{i}j} / \partial p_j d\mu}  = - \int_{i \in \mathcal{S}_k} \left( \frac{\partial e_{ij} / \partial p_j}{\int_{\tilde{i} \in \mathcal{S}_j} \partial e_{\tilde{i}j} / \partial p_j d\mu}  \right) \left( - \frac{\partial e_{ik}/ \partial p_j}{\partial e_{ij} / \partial p_j}  \right) d\mu $, which is what I wanted to show.
\end{proof}

Assumption \ref{assumption:discrete.continuous.demand} yields the following characterization of revenue diversion ratios.

\begin{lemma}[Revenue diversion ratios with discrete-continuous demand] \label{lemma:revenue.diversion.ratio.with.discrete.continuous.demand} Under Assumption \ref{assumption:discrete.continuous.demand}, the revenue diversion ratios for $j \neq k$ can be written as
\[
D_{j \to k}^R = \int_{i \in \mathcal{S}_k} w_{ij}D_{i,j \to k}^e d\mu, \; \text{ where } w_{ij} =  \frac{\partial \alpha_{ij} / \partial u_{ij} B_i}{\int_{\tilde{i} \in \mathcal{S}_j} \partial \alpha_{\tilde{i}j}/\partial u_{\tilde{i}j} B_{\tilde{i}} d\mu}, \; D_{i, j \to k}^R = - \frac{\partial \alpha_{ik} / \partial u_{ij}}{\partial \alpha_{ij} / \partial u_{ij}}.
\]
\end{lemma}
\begin{proof}
    The results follow by observing that Assumption \ref{assumption:discrete.continuous.demand} implies, for an arbitrary $j$ and $k$, $\frac{\partial e_{ik}}{\partial p_j} = \frac{\partial \alpha_{ik}}{\partial u_{ij}} \frac{\partial u_{ij}}{\partial p_j} B_i = \frac{\partial \alpha_{ik}}{\partial u_{ij}} \frac{1-\eta}{p_j} B_i$. Plugging in the expressions into the equation in Lemma \ref{lemma:revenue.diversion.ratio.as.a.weighted.sum.general.case} and canceling out $\frac{1-\eta}{p_j}$ and $B_i$ gives the desired expressions.
\end{proof}

Lemma \ref{lemma:revenue.diversion.ratio.with.discrete.continuous.demand} shows that estimating the revenue diversion ratios requires only the sensitivity of consumers' expenditure shares with respect to the mean utility of product $j$ at the margin. Specifically, if the analyst can observe the partial derivatives $\partial \alpha_{ik} / \partial u_{ij}$ at the pre-merger equilibrium, the revenue diversion ratios can be estimated. This estimation relies on the classical conditions for applying the inversion theorem of \citet{hotz1993conditional}, which ensure that the mapping from utility indices $u_i = (u_{ij})_{j \in \mathcal{C}_i}$ to expenditure shares $\alpha_i = (\alpha_{ij})_{j \in \mathcal{C}_i}$ is invertible.\footnote{I consider \citet{hotz1993conditional}'s classical assumptions since they hold in most applications. However, recent developments allow for weaker assumption; see, e.g., \citet{chiong2016duality}, \citet{galichon2018optimal}, and \citet{sorensen2022mcfadden}.} Given that the analyst knows this inverse mapping and observes pre-merger values of $\alpha_i$, the relevant partial derivatives in Lemma \ref{lemma:revenue.diversion.ratio.with.discrete.continuous.demand} can be evaluated.

\begin{proposition}[Identification of revenue diversion ratios under discrete-continuous demand] \label{proposition:identification.of.revenue.diversion.ratios.under.discrete.continuous.demand} Suppose that Assumption \ref{assumption:discrete.continuous.demand} holds. If the analyst observes $\alpha_i = (\alpha_{ij})_{j \in \mathcal{C}_i}$ for all consumer $i$ and knows the distribution $H_i$ of the idiosyncratic utility shocks $\varepsilon_{ij}$, then the revenue diversion ratios for all product pairs in the market are identified.
\end{proposition}
\begin{proof}
    Since consumer budgets $B_i$ are known by assumption, it suffices to identify the values of $\partial \alpha_{ik}/\partial u_{ij}$ for all consumers $i$ and product pairs $(j,k)$. Given knowledge of the distribution $H_i$, which is assumed to be absolutely continuous, the mapping from utility indices $u_i$ to expenditure shares $\alpha_i$ is invertible \citep{hotz1993conditional}. Therefore, $\alpha_i$ can be expressed as a known function of $u_i$. Because the analyst observes the pre-merger expenditure shares $\alpha_{i}$, the corresponding utility indices $u_i$ can be recovered through inversion. As a result, the analyst can compute all relevant partial derivatives $\partial \alpha_{ik} / \partial u_{ij}$.
\end{proof}

While the above proposition assumes knowledge of budgets, expenditure shares, and preference shock distributions, these can also be estimated parametrically by specifying $B_i = B_i^\theta$, $\alpha_i = \alpha_i^\theta$, and $H_i = H_i^\theta$, and estimating $\theta$ using moment conditions. This approach requires rich covariates to control for unobserved prices. For example, \citet{ellickson2020measuring} assume: (i) grocery budgets are a fixed share of income; (ii) utility is a linear function of distance, tract, and store characteristics; and (iii) expenditure shares follow a nested logit model with nests defined by store formats. They estimate parameters using store-level revenue data and control for unobserved prices with chain indicators, citing evidence of within-chain price uniformity \citep{dellavigna2019uniform, hitsch2021prices}.

\subsection{Merger Simulation \label{subsection:merger.simulation.under.discrete.continuous.demand}}

The discrete-continuous demand assumption described in Assumption  \ref{assumption:discrete.continuous.demand} also facilitates merger simulation. The following lemma is the discrete-continuous demand analog of Lemma \ref{lemma:merger.simulation.under.ces}. 

\begin{lemma}[Merger simulation under discrete-continuous demand] \label{lemma:merger.simulation.under.discrete.continuous.demand} Suppose Assumption \ref{assumption:discrete.continuous.demand} holds. For each $j \in \mathcal{J}$, let $\ddot{p}_j = (p_j^\text{post} - p_j^\text{pre})/p_j^\text{pre}$, where $p_j$ is an arbitrary (possibly non-equilibrium) candidate of post-merger price. Then, the corresponding utility indices, consumer expenditure shares, own-price elasticities of demand, revenue diversion ratios, and margins induced by $\ddot{p} = (\ddot{p}_j)_{j \in \mathcal{J}}$ can be expressed as
\begin{align*}
    u_{ij}^\text{post} & = u_{ij}^\text{pre} + (1-\eta) \log (1 + \ddot{p}_j), \\
    \alpha_{ij}^\text{post} & = \mathbb{E}[\{j = \arg \max_{l \in \mathcal{C}_i} (u_{il}^\text{post} + \varepsilon_{il})\}], \\
    \epsilon_{jj}^{\text{post}} & = \frac{1-\eta}{R_j^\text{post}} \int_{i \in \mathcal{S}_j} (\partial \alpha_{ij}^\text{post} / \partial u_{ij}^\text{post}) B_i d\mu - 1, \\
    D_{j \to k}^{R, \text{post}} & = \int_{i \in \mathcal{S}_k} w_{ij}^\text{post} D_{i,j \to k}^{e, \text{post}} d\mu, \\
    m_j^\text{post} & = 1 - (1 - m_j^\text{pre})\left( \frac{1 + \ddot{c}_j}{1 + \ddot{p}_j} \right),
\end{align*}
where $w_{ij}^\text{post} = \frac{\partial \alpha_{ij}^\text{post} / \partial u_{ij}^\text{post} }{\int_{\tilde{i} \in \mathcal{S}_j} \partial \alpha_{\tilde{i}j}^\text{post} / \partial u_{\tilde{i}j}^\text{post}d\mu}$ and $D_{i,j \to k}^{e,\text{post}} = - \frac{\partial \alpha_{ik}^\text{post} / \partial u_{ij}^\text{post}}{\partial \alpha_{ij}^\text{post} / \partial u_{ij}^\text{post}}$. 
\end{lemma}

To express the objects in Lemma \ref{lemma:merger.simulation.under.discrete.continuous.demand} as known functions of $\ddot{p}$, it suffices to establish the conditions under which the Hotz-Miller inversion theorem applies. These data requirements are the same as those in the CES case (Proposition \ref{proposition:merger.simulation.under.ces}), with the additional condition that the analyst must know the distribution $H_i$ to implement the inversion.

\begin{proposition}[Merger simulation under discrete-continuous demand] \label{proposition:merger.simulation.under.discrete.continuous.demand}

Suppose Assumption \ref{assumption:discrete.continuous.demand} holds. If the analyst observes consumer expenditure shares $(\alpha_{ij})_{i \in \mathcal{I}, j \in \mathcal{J}}$, product margins $(m_j)_{j \in \mathcal{J}}$, and percentage reductions in marginal costs $(\ddot{c}_j)_{j \in \mathcal{J}}$, and knows the distribution $H_i$ of the idiosyncratic utility shocks $\varepsilon_{ij}$, then the percentage price changes from the pre-merger to the post-merger equilibrium are identified for all products in the market.
\end{proposition}
\begin{proof}
The result follows from a sequence of observations. First, given the stated conditions, the pre-merger utility indices $(u_{ij}^\text{pre})_{i \in \mathcal{I}, j \in \mathcal{J}}$ can be identified by inverting the mapping from $u_i$ to $\alpha_i$ using the Hotz-Miller inversion theorem. Second, Proposition \ref{proposition:identification.of.revenue.diversion.ratios.under.discrete.continuous.demand} establishes that the revenue diversion ratios between all product pairs are identified. Third, $\eta$ can be identified after identifying the own-price elasticities of demand and using $\epsilon_{jj}^\text{pre} = \frac{1-\eta}{R_j^\text{pre}} \int_{i \in \mathcal{S}_j} (\partial \alpha_{ij}^\text{pre} / \partial u_{ij}^\text{pre}) B_i d\mu - 1$. Fourth, the post-merger utility indices $u_{ij}^\text{post}$ are known functions of $\ddot{p}_j$, since $u_{ij}^\text{pre}$ and the parameter $\eta$ are known. Fifth, given $u_i^\text{post}$, the post-merger expenditure shares $\alpha_{ij}^\text{post}$ become known functions of $\ddot{p}$. Sixth, the post-merger own-price elasticity $\epsilon_{jj}^\text{post}$ depends on $R_j^\text{post}$ and $\partial \alpha_{ij}^\text{post} / \partial u_{ij}^\text{post}$, both of which are known functions of $u_{ij}^\text{post}$, and thus of $\ddot{p}$. Seventh, the post-merger revenue diversion ratio $D_{j \to k}^{R, \text{post}}$ depends on $w_{ij}^\text{post}$ and $D_{i,j \to k}^{e,\text{post}}$, both of which are functions of $u_{ij}^\text{post}$, and therefore of $\ddot{p}$. Finally, the post-merger margin $m_j^\text{post}$ is a closed-form function of the pre-merger margin $m_j^\text{pre}$, the percentage cost reduction $\ddot{c}_j$, and the percentage price change $\ddot{p}_j$. Since $m_j^\text{pre}$ and $\ddot{c}_j$ are observed, $m_j^\text{post}$ is also a known function of $\ddot{p}_j$. Together, these results imply that the post-merger first-order conditions can be expressed as a system $f(\ddot{p}) = 0$, which can be solved for the vector of percentage price changes $\ddot{p}$.
\end{proof}

% \clearpage
% \input{S.Z.Misc/S.Z.Misc.0}

% \part*{Supplementary Material (Work in Progress)}
% \input{S.Y.Implementation.of.Empirical.Applications/S.Y.0.Implementation.of.Empirical.Application}

\clearpage

\setcounter{page}{1}  % Reset page numbering
\renewcommand{\thepage}{A\arabic{page}}  % Page numbers start with "A"
\setcounter{section}{0} % Reset section numbering

\begin{center}
    \LARGE Online Appendix for \\
    \LARGE{Merger Analysis with Unobserved Prices} \\
    \vspace{1cm}
    \normalsize Paul S. Koh \\
    \normalsize August 30, 2025
\end{center}

\section{Review of Unilateral Effects Analysis \label{section:B.review.of.unilateral.effects.analysis}}
This section reviews the basics of first-order unilateral effects analysis, specifically firms' profit maximization conditions, gross upward pricing pressure indices, and compensating marginal cost reductions.

\subsection{Firm's Problem \label{section:B.1.firms.problem}}
The multiproduct firms engage in a Bertrand-Nash pricing game. Each firm $F \in \mathcal{F}$ maximizes its total profit $\sum_{j \in \mathcal{J}_F} \pi_j$ with respect to a vector of prices $(p_j)_{j \in \mathcal{J}_F}$. The first-order condition with respect to $p_j$ is $q_j + (p_j - c_j) \frac{\partial q_j}{\partial p_j} + \sum_{l \in \mathcal{J}_F \backslash j} (p_l - c_l) \frac{\partial q_l}{\partial p_j} = 0$. Normalizing the FOC to be quasilinear in the marginal cost gives
\[
- p_j \epsilon_{jj}^{-1} - (p_j - c_j) + \sum_{l \in \mathcal{J}_F \backslash j } (p_l - c_l) D_{j \to l} = 0,
\]
where $\epsilon_{jj} \equiv \frac{\partial q_j}{\partial p_j} \frac{p_j}{q_j}$ is the own-price elasticity of demand, and $D_{j \to l} \equiv - \frac{\partial q_l / \partial p_j}{\partial q_j / \partial p_j}$ is the quantity diversion ratio from product $j$ to product $l$.\footnote{Rearranging the first-order conditions gives the optimal markup equation $p_j (1 + \epsilon_{jj}^{-1}) = c_j + \sum_{l \in \mathcal{J}_F \backslash j} (p_l - c_l)D_{j \to l}$, which in turn implies that $(1 + \epsilon_{jj}^{-1}) > 0$, or, equivalently, $\varepsilon_{jj} < -1$, i.e., in a Bertrand-Nash equilibrium, firms always price at the elastic region of demand.} Further dividing the FOC by $p_j$ normalizes the first-order conditions to be quasilinear in margins and gives \eqref{equation:optimal.pricing.equation.2}: 
\[
-\epsilon_{jj}^{-1} - m_j + \sum_{l \in \mathcal{J}_F \backslash j} m_l D_{j \to l} \frac{p_l}{p_j} = 0.
\]

\subsection{Upward Pricing Pressure \label{section:B.2.upward.pricing.pressure}}
Under unilateral effects theory, a merger raises the merging firms' pricing incentives by internalizing the diversion of consumers to each other's products. \citet{farrell2010antitrust} propose measuring this incentive using \emph{upward pricing pressure (UPP)}, defined as the difference between the pre- and post-merger first-order conditions, normalized to be quasilinear in marginal cost and evaluated at pre-merger prices.

Formally, consider a merger between two firms $A$ and $B$. The upward pricing pressure associated with product $j$ of firm $A$ is defined as
\begin{equation*}%\label{equation:upward.pricing.pressure}
    \mathit{UPP}_j \equiv \Delta c_j + \sum_{k \in \mathcal{J}_B} (p_k - c_k) D_{j \to k},
\end{equation*}
where $\Delta c_j = c_j^\text{post} - c_j^\text{pre}<0$ represents the reduction in marginal cost due to merger-specific efficiencies; those of firm $B$ are defined symmetrically.\footnote{I occasionally omit the superscript ``pre'' when it is clear the object is that of the pre-merger equilibrium.} 

The gross upward pricing pressure index is a unit-free measure of upward pricing pressure and is obtained by normalizing upward pricing pressure by price, i.e., $\mathit{GUPPI}_j \equiv \mathit{UPP}_j / p_j$. Using the expression for the UPP above, GUPPI can be rewritten as \eqref{equation:gross.upward.pricing.pressure}:
\[
\mathit{GUPPI}_j = \ddot{c}_j (1-m_j) + \sum_{k \in \mathcal{J}_B} m_k D_{j \to k} \frac{p_k}{p_j},
\]
where $\ddot{c}_j \equiv  \Delta c_j / c_j^\text{pre} \in (-1,0)$ represents the percentage decrease in marginal cost from the pre-merger equilibrium.

\subsection{Compensating Marginal Cost Reduction \label{section:B.3.cmcr}}
\emph{Compensating marginal cost reductions (CMCR)} are defined as the percentage decrease in the merging parties' costs that would leave the pre-merger prices unchanged after the merger. \citet{werden1996robust} shows that CMCRs are identified when prices, margins, and diversion ratios are observed with single-product firms.

In contrast to upward pricing pressures that assume other products' marginal costs are fixed, compensating marginal cost reductions capture the effects of simultaneous changes in marginal costs. Upward pricing pressure can be more conservative because the value of sales diverted to the merging counterparty is larger if profit margins increase due to a reduction in marginal costs. \citet{farrell2010antitrust} endorses upward pricing pressure on the grounds of simplicity and transparency while acknowledging that compensating marginal cost reductions can be more accurate.

Let $\ddot{c}_j \equiv \frac{c_j^\text{post} - c_j^\text{pre}}{c_j^\text{pre}}$ be the compensating marginal cost reductions, where $c_j^\text{post}$'s are pinned by finding a vector of post-merger marginal costs that offset any upward pricing incentives. Rewriting the CMCR in terms of pre- and post-merger margins gives
\begin{equation}\label{equation:cmcr.as.function.of.margins}
\ddot{c}_j = \frac{m_j^0 - m_j^1}{1 - m_j^0}.    
\end{equation}
To see that \eqref{equation:cmcr.as.function.of.margins} holds, note that, since $p_j^1 = p_j^0$, $m_j^0 - m_j^1 = \frac{1}{p_j^0} ((p_j^0 - c_j^0) - (p_j^0 - c_j^1)) = \frac{\Delta c_j}{p_j^0}$. Then
\[
\frac{m_j^0 - m_j^1}{1 - m_j^0} = \frac{\frac{\Delta c_j}{p_j}}{\frac{p_j^0}{p_j^0} - \frac{p_j^0 - c_j^0}{p_j^0}} = \frac{\Delta c_j / p_j^0}{c_j^0 / p_j^0} = \frac{\Delta c_j}{c_j^0} = \ddot{c}_j.
\]
The post-merger margins are defined by the post-merger first-order conditions
\begin{equation*}\label{equation:post.merger.foc.for.cmcr.original}
    - \epsilon_{jj}^{-1} - m_j^1 + \sum_{l \in \mathcal{J}_A \backslash j} m_l^1 D_{j \to l}\frac{p_l}{p_j} + \sum_{k \in \mathcal{J}_B } m_k^1 D_{j \to k} \frac{p_k}{p_j} = 0.
\end{equation*}
Finding post-merger margins amounts to solving a linear system of equations. After finding the vector of post-merger margins $(m_j^1)_{j \in \mathcal{J}_A \cup \mathcal{J}_B}$ that solve the system of equations, the analyst can plug them into \eqref{equation:cmcr.as.function.of.margins} to obtain the CMCR.

% \begin{remark}
% Setting $\mathit{GUPPI}_j = 0$ and solving for $\psi_j$ gives the critical value of cost-reduction for product $j$:
% \begin{equation}\label{equation:mc.reduction.to.offset.upp}
%     \psi_j^* = \left( 1 - m_j \right)^{-1} \left( \sum_{k \in \mathcal{J}_g} m_k D_{jk} \frac{p_k}{p_j} \right).
% \end{equation}
% The analyst can report \eqref{equation:mc.reduction.to.offset.upp} to gauge the reduction in marginal costs required to offset the upward pricing pressure. Note, however, that \eqref{equation:mc.reduction.to.offset.upp} differs from CMCR. Upward pricing pressure assumes other firms' marginal costs remain fixed. In contrast, CMCR assumes the costs decrease simultaneously and solves simultaneous equations.\footnote{\citet{farrell2010antitrust} endorses upward pricing pressure on grounds of simplicity and transparency while acknowledging that compensating marginal cost reductions can be more accurate.}
% \end{remark}

\section{Identification of Compensating Variation Under CES Preferences \label{sec:identification.of.compensating.variation.under.ces.preferences}}

In this section, I show that the compensating variation associated with a vector of price increases is point-identified under the CES utility assumption. While the previous sections presented an approximation to consumer surplus variation, the CES structure allows for an exact expression for compensating variation in response to a vector of simultaneous price changes.

Let $V_i(p,B_i)$ be consumer $i$'s indirect utility at price $p$ and budget $B_i$. The compensating variation for consumer $i$ associated with a price increase from $p^0$ to $p^1$ is defined as $\mathit{CV}_i \in \mathbb{R}$ that solves $V_i(p^0,B_i) = V_i(p^1, B_i-\mathit{CV}_i)$. The total compensating variation is $\mathit{CV} = \int_{i \in \mathcal{I}} \mathit{CV}_i d\mu$. Under the CES utility assumption, the compensating variation admits a closed-form expression.

\begin{lemma}[Compensating variation under CES utility] \label{lemma:compensating.variation}
Suppose Assumption \ref{assumption:ces.utility} holds. Let $P_i(u_i) \equiv (\sum_{k \in \mathcal{C}_i} \exp (u_{ik}))^{\frac{1}{\eta-1}}$. For each consumer $i$, the compensating variation associated with a price change from $p^0$ to $p^1$ is
\begin{equation}
    \mathit{CV}_i = B_i \left(1 - \frac{P_i(u_i^{0})}{P_i(u_i^{1})} \right)
\end{equation}
where $u_i^0 = u_i(p^0)$ and $u_i^1 = u_i(p^1)$. Furthermore, $u_{ij}^1 = u_{ij}^0 + (1-\eta) \log(1 + \ddot{p}_j)$.
\end{lemma}
\begin{proof}
Consumer $i$'s indirect utility function is $V_i(p,B_i) = AP_i^{\beta,\eta}(p)B_i$, where $P_i^{\beta,\eta}(p) \equiv \left( \sum_{k \in \mathcal{C}_i } \beta_{ik} p_k^{1-\eta} \right)^{\frac{1}{\eta-1}}$. Solving $V_i(p^0, B_i) = V_i(p^1, B_i - \mathit{CV}_i)$ for $\mathit{CV}_i$ gives the desired expression.
\end{proof}

The above lemma shows that Assumptions \ref{assumption:data}--\ref{assumption:budget.and.consumer.expenditure.data} are sufficient to identify the compensating variation. To see this, note that since $u_{ij} = \log \beta_{ij} + (1-\eta) \log p_j$, $\exp(u_{ij}) = \beta_{ij} p_j^{1-\eta}$, so the compensating variation for consumer $i$ reduces to:
\[
\mathit{CV}_i = B_i \left( 1 - \left( \frac{\sum_{k \in \mathcal{C}_i} \exp(u_{ik}^0)}{\sum_{k \in \mathcal{C}_i } \exp(u_{ik}^1)} \right)^{\frac{1}{1-\eta}} \right),
\]
where $u_{ij}^1 = u_{ij}^0 + (1-\eta) \log (1 + \ddot{p}_j)$. Thus, the compensating variation with respect to a given $\ddot{p}$ is identified if the pre-merger utility indices $u_{ij}^0$'s of the products associated with the price increases and the parameter $\eta$ are known. The pre-merger utility indices can be obtained using the inversion $\log \alpha_{ik} - \log \alpha_{i0} = u_{ik}$, and $\eta$ can be identified from estimated own-price elasticities and consumer expenditure shares (see equation \eqref{equation:ces.own.price.elasticity.of.revenue}).

\section{Pass-Through Matrix Under CES Preference \label{section:pass.through.matrix.under.ces.preference}}

In this section, I derive the merger pass-through matrix assuming CES preference to illustrate that the analyst can estimate pass-through rates using the merging firms' data used for GUPPI calculation. For simplicity, I assume a single representative consumer and single-product firms. I apply the results to calculate the merger pass-through matrix for the Staples/Office Depot example in Section \ref{section:6.empirical.application}.

\subsubsection*{Definition of Merger Pass-through Matrix}
Consider a merger between firms $j$ and $k$. Let
\begin{align*}
    h_j(\tilde{p}) & \equiv -\varepsilon_{jj}^{-1} - m_j + (1 + \varepsilon_{jj}^{-1}) m_k D_{j \to k}^R, \\
    h_k(\tilde{p}) & \equiv -\varepsilon_{kk}^{-1} - m_k + (1 + \varepsilon_{kk}^{-1}) m_j D_{k\to j}^R, 
\end{align*}
where $\tilde{p}_j \equiv \log p_j$ for each $j \in \mathcal{J}$ so that $h(\tilde{p}^{\text{post}}) = 0$ describe the merging firms' post-merger first-order conditions. The above first-order conditions are normalized to be quasilinear in margins and expressed as functions of log prices so that we can calculate how GUPPIs translate to merger price effects in percentage change terms.

A merger pass-through matrix describes the marginal effects of (normalized) tax rates on merged firms' prices at the pre-merger equilibrium. To formalize the definition, first evaluate $h(\tilde{p})$ at the pre-merger prices to get $h(\tilde{p}^\text{pre}) = c$ for some vector of constants $c$. Introduce tax rates $\tilde{t}$ and assume that $\tilde{p} = \tilde{p}(\tilde{t})$ such that $h(\tilde{p}) + \tilde{t} = c$ around $\tilde{p}^\text{pre}$. By the implicit function theorem,
\[
\frac{\partial \tilde{p}}{\partial \tilde{t}} \cdot \frac{\partial h(\tilde{p})}{\partial \tilde{p}}  + I = 0.
\]
Rearranging the above gives the merger pass-through matrix defined as 
\[
M \equiv \frac{\partial{\tilde{p}}}{\partial \tilde{t}} \bigg \vert_{\tilde{t} = 0} = - \left( \frac{\partial h(\tilde{p})}{\partial \tilde{p}} \right)^{-1} \bigg \vert_{\tilde{p} = \tilde{p}^\text{pre}}
\]
Then, the first-order approximation of merger price effects is $\ddot{p} \approx M \cdot \mathit{GUPPI}$.

\subsubsection*{Derivation}
Recall that under CES preference, $\varepsilon_{jj}^R = (1-\alpha_j) (1-\eta)$ and $D_{j \to k}^R = \frac{\alpha_k}{1 - \alpha_j}$. Then since $\varepsilon_{jj} = \varepsilon_{jj}^R - 1$, we have $\varepsilon_{jj} = (1-\alpha_j)(1-\eta) - 1$. In addition, $\partial \alpha_j / \partial u_j = \alpha_j (1-\alpha_j)$, and $\partial \alpha_j / \partial u_k = -\alpha_k \alpha_j$ if $j \neq k$. Moreover, $\partial u_j / \partial \tilde{p}_j = (1-\eta)$ for all product $j \in \mathcal{J}$.

Let us derive the $2 \times 2$ merger pass-through matrix $M$ given by
\[
M = - 
\begin{bmatrix}
    \frac{\partial h_j}{\partial \tilde{p}_j} & \frac{\partial h_j}{\partial \tilde{p}_k} \\
    \frac{\partial h_k}{\partial \tilde{p}_j} & 
    \frac{\partial h_k}{\partial \tilde{p}_k}
\end{bmatrix}^{-1}
.
\]
The derivatives of $h_j$ with respect to $\tilde{p}_j$ and $\tilde{p}_k$ are given by
\begin{align*}
    \frac{\partial h_j}{\partial \tilde{p}_j} &= - \frac{\partial \varepsilon_{jj}^{-1}}{\partial \tilde{p}_j} - \frac{\partial m_j}{\partial \tilde{p}_j} + \frac{\partial (1 + \varepsilon_{jj}^{-1})}{\partial \tilde{p}_j} m_k D_{j \to k}^R + (1 + \varepsilon_{jj}^{-1}) m_k \frac{\partial D_{j \to k}^R}{\partial \tilde{p}_j} \\ 
    \frac{\partial h_j}{\partial \tilde{p}_k} &= - \frac{\partial \varepsilon_{jj}^{-1}}{\partial \tilde{p}_k} + \frac{\partial (1 + \varepsilon_{jj}^{-1})}{\partial \tilde{p}_k} m_k D_{j \to k}^R + (1 + \varepsilon_{jj}^{-1}) \frac{\partial m_k}{\partial \tilde{p}_k} D_{j \to k}^R + (1 + \varepsilon_{jj}^{-1}) m_k \frac{\partial D_{j \to k}^R}{\partial \tilde{p}_k}.
\end{align*}
The partial derivatives $\frac{\partial h_k}{\partial \tilde{p}_j}$ and $\frac{\partial h_k}{\partial \tilde{p}_k}$ are obtained symmetrically.

First, since $\frac{\partial \varepsilon_{jj}}{\partial \tilde{p}_j} = ( - \frac{\partial \alpha_j}{\partial u_j} \frac{\partial u_j}{\partial \tilde{p}_j}) (1-\eta)$ and $\frac{\partial \varepsilon_{jj}}{\partial \tilde{p}_k} = ( - \frac{\partial \alpha_j}{\partial u_k} \frac{\partial u_k}{\partial \tilde{p}_k})(1-\eta)$,
\begin{align*}
    \frac{\partial \varepsilon_{jj}}{\partial \tilde{p}_j} &= (-1) \alpha_j (1-\alpha_j) (1-\eta) (1-\eta), \\
    \frac{\partial \varepsilon_{jj}}{\partial \tilde{p}_k} &=  \alpha_j \alpha_k (1-\eta)(1-\eta) .  
\end{align*}
Second, since $\frac{\partial m_j}{\partial \tilde{p}_j} = \frac{\partial (1 - c_j/\exp(\tilde{p}_j))}{\partial \tilde{p}_j} = c_j / \exp(\tilde{p}_j) = c_j / p_j$,
\begin{equation*}
    \frac{\partial m_j}{\partial \tilde{p}_j} = 1 - m_j.
\end{equation*}
Third, from $D_{j \to k}^R = \frac{\alpha_k}{1-\alpha_j}$, we have
\begin{align*}
    \frac{\partial D_{j \to k}^R}{\partial \tilde{p}_j} & = \frac{\partial \alpha_k}{\partial u_j} \frac{\partial u_j}{\partial \tilde{p}_j} \frac{1}{1-\alpha_j} + \alpha_k (-1) (1-\alpha_j)^{-2} (-1) \frac{\partial \alpha_j}{\partial u_j} \frac{\partial u_j}{\partial \tilde{p}_j} \\
    & = -\alpha_k \alpha_j (1-\eta) \frac{1}{1-\alpha_j} + \frac{\alpha_k}{(1-\alpha_j)^2} \alpha_j (1-\alpha_j) (1-\eta) \\
    & = -\alpha_j (1-\eta) D_{j \to k}^R + \alpha_j (1-\eta) D_{j \to k}^R \\ 
    & = 0.
\end{align*}
Next, 
\begin{align*}
    \frac{\partial D_{j \to k}^R}{\partial \tilde{p}_k} & = \frac{\partial \alpha_k}{\partial u_k} \frac{\partial u_k}{\partial \tilde{p}_k} \frac{1}{1-\alpha_j} + \alpha_k (-1) (1-\alpha_j)^{-2} (-1) \frac{\partial \alpha_j}{\partial u_k} \frac{\partial u_k}{\partial \tilde{p}_k} \\ 
    & = \alpha_k (1-\alpha_k) (1-\eta) \frac{1}{1-\alpha_j} + \frac{\alpha_k}{(1-\alpha_j)^2} (-\alpha_j \alpha_k) (1-\eta) \\ 
    & = (1-\alpha_k) (1-\eta) D_{j \to k}^R - \alpha_j (1-\eta) (D_{j \to k}^R)^2 \\
    & = \alpha_j  (1-\eta) D_{k \to j}^{R,-1} D_{j \to k}^R - \alpha_j (1-\eta) (D_{j \to k}^R)^2 \\
    & = \alpha_j (1-\eta) D_{j \to k}^R (D_{k \to j}^{R,-1} - D_{j \to k}^R)
\end{align*}
Thus,
\begin{align*}
    \frac{\partial \varepsilon_{jj}^{-1}}{\partial \tilde{p}_j} & = \varepsilon_{jj}^{-2} \alpha_j (1-\alpha_j) (1-\eta) (1-\eta), \\ 
    \frac{\partial \varepsilon_{jj}^{-1}}{\partial \tilde{p}_k} & = (-1) \varepsilon_{jj}^{-2} \alpha_j \alpha_k (1-\eta)(1-\eta), \\ 
    \frac{\partial m_j}{\partial \tilde{p}_j} &= 1-m_j, \\
    \frac{\partial D_{j \to k}^R}{\partial \tilde{p}_j} & = 0, \\
    \frac{\partial D_{j \to k}^R}{\partial \tilde{p}_k} & = \alpha_j (1-\eta) D_{j \to k}^R (D_{k \to j}^{R,-1} - D_{j \to k}^R).
\end{align*}
Plugging the above to $\partial h_j / \partial \tilde{p}_j$ and $\partial h_j / \partial \tilde{p}_k$ gives
\begin{align*}
    \frac{\partial h_j}{\partial \tilde{p}_j} & = (-1) \frac{1}{\varepsilon_{jj}^2} \alpha_j (1-\alpha_j)(1-\eta)^2(1 - m_k D_{j \to k}^R) - (1- m_j), \\
    \frac{\partial h_j}{\partial \tilde{p}_k} & = \frac{1}{\varepsilon_{jj}^2} \alpha_k \alpha_j (1-\eta)^2(1- m_k D_{j \to k}^R) + (1 + \varepsilon_{jj}^{-1})(1 - m_k) D_{j \to k}^R  \\ 
    & + (1 + \varepsilon_{jj}^{-1})m_k \alpha_j (1-\eta) D_{j \to k}^R (D_{k \to j}^{R,-1} - D_{j \to k}^R).
\end{align*}
The derivatives $\frac{\partial h_k}{\partial \tilde{p}_j}$ and $\frac{\partial h_k}{\partial p_k}$ are obtained symmetrically.

\subsubsection*{Application to the Staples/Office Depot Example}
In the Staples/Office Depot example, we had $\alpha_1 = 0.473$, $\alpha_2 = 0.316$, $m_1 = 0.258$, $m_2 = 0.234$, $\varepsilon_{11} = - 3.875$, $\varepsilon_{22} = -4.273$, $D_{1 \to 2}^R = 0.599$, $D_{2 \to 1}^R = 0.691$, and $\eta = 6.121$. The GUPPIs were $\mathit{GUPPI}_1 = 0.104$ and $\mathit{GUPPI}_2 = 0.137$. 

The merger pass-through matrix calculated with the above formula is
\[
M = 
\begin{bmatrix}
1.005 & 0.345 \\
0.347 & 1.098
\end{bmatrix}.
\]
Thus, the first-order approximation gives
\[
\ddot{p} \approx M \cdot \mathit{GUPPI} = 
\begin{bmatrix}
    0.152 \\ 0.187
\end{bmatrix}
.
\]

\section{Accuracy of Upward Pricing Pressure \label{section:accuracy.of.upward.pricing.pressure}}

In one of my empirical applications, I have used \begin{equation}\label{equation:upward.pricing.pressure.predicts.merger.price.effects}
    \Delta p_j \approx \mathit{UPP}_j.
\end{equation} 
to predict merger price effects without separately estimating merger pass-through rates. In a companion paper \citep{koh2024concentration}, I repeat the simulation exercise in \citet{miller2017upward} and examine whether \eqref{equation:upward.pricing.pressure.predicts.merger.price.effects} is reasonable. \citet{miller2017upward} does not study the case of CES demand. 

\begin{figure}[htbp!]
\centering
\begin{subfigure}{.5\textwidth}
  \centering
  \includegraphics[width=.6\linewidth]{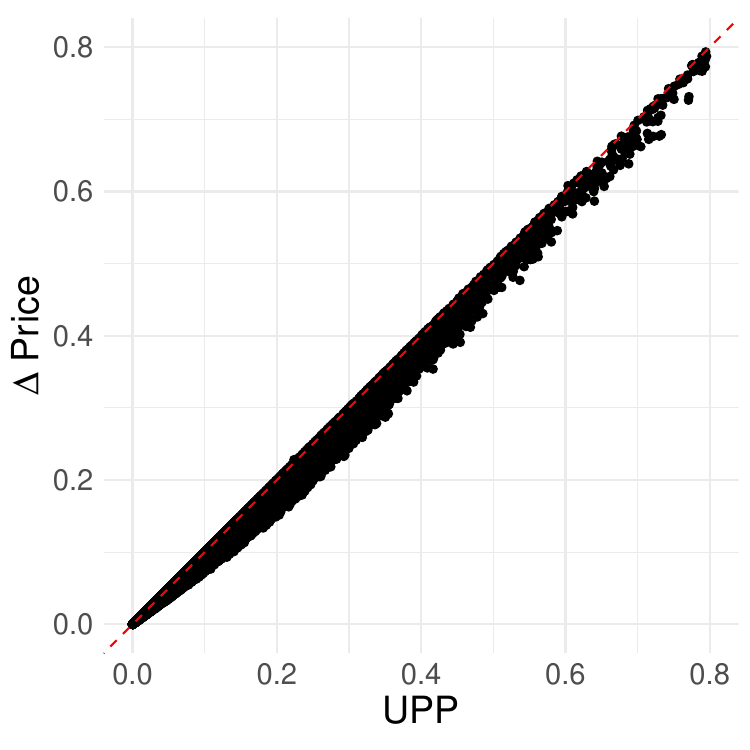}
  \caption{Logit}
  \label{fig:sub1}
\end{subfigure}%
\begin{subfigure}{.5\textwidth}
  \centering
  \includegraphics[width=.6\linewidth]{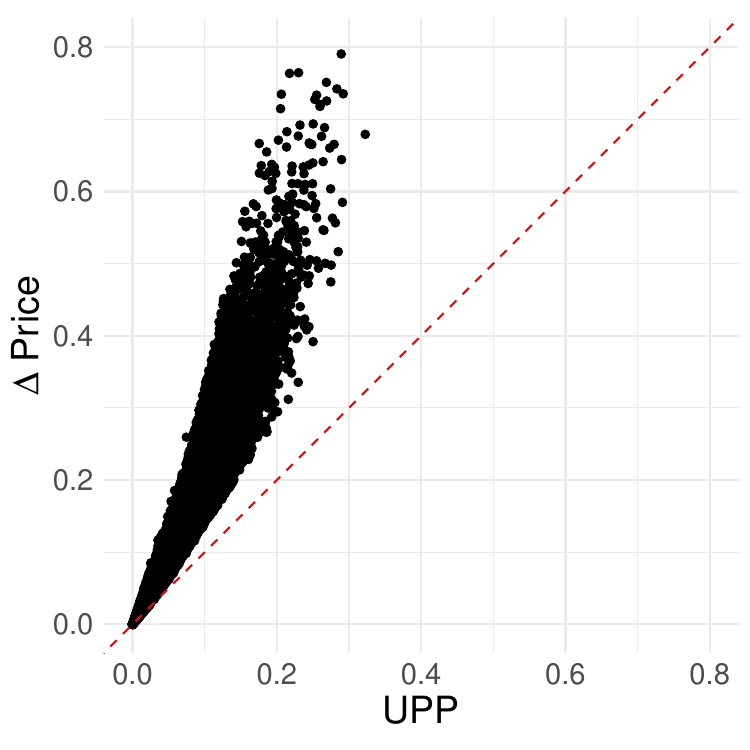}
  \caption{CES}
  \label{fig:sub2}
\end{subfigure}
\caption{Accuracy of UPP}
\label{figure:accuracy.of.upp}
\end{figure}

I report the simulation results in Figure \ref{figure:accuracy.of.upp}. First, Figure \ref{figure:accuracy.of.upp}-(a) shows that upward pricing pressure accurately predicts merger price effects in the case of logit demand. The result is also consistent with \citet{miller2017upward} (see their Figure 2). Next, Figure \ref{figure:accuracy.of.upp}-(b) shows the simulation results with CES demand. It shows that upward pricing pressure underpredicts merger price effects. Thus, using upward pricing pressures as proxies for merger price effects would yield conservative predictions.

\end{document}